\documentclass[10pt,conference]{IEEEtran}
\IEEEoverridecommandlockouts

\usepackage{amsmath,amssymb,amsfonts}
\usepackage{textcomp}
\usepackage[absolute]{textpos}

\usepackage{multirow}
\usepackage{graphicx}
\usepackage{pdfpages}
\usepackage{ulem}
\usepackage{hyperref}
\usepackage{booktabs} 
\usepackage{siunitx} 
\usepackage{etoolbox}

\usepackage[switch]{lineno} 

\usepackage[linesnumbered,titlenumbered,ruled,vlined,resetcount,algosection]{algorithm2e}
\usepackage{xcolor}
\usepackage{color}
\def\BibTeX{{\rm B\kern-.05em{\sc i\kern-.025em b}\kern-.08em
    T\kern-.1667em\lower.7ex\hbox{E}\kern-.125emX}}

\usepackage{amsthm}
\newtheorem{theorem}{Theorem}
\newtheorem{corollary}{Corollary}

\usepackage{hhline}
\usepackage{colortbl}
\usepackage{threeparttable}
\usepackage{subfig}
\usepackage{comment}
\usepackage{enumitem}
\usepackage{tikz}
\newcommand*\circled[1]{\tikz[baseline=(char.base)]{
            \node[shape=circle,fill,inner sep=1.3pt] (char) {\textcolor{white}{#1}};}}

\def\BibTeX{{\rm B\kern-.05em{\sc i\kern-.025em b}\kern-.08em
    T\kern-.1667em\lower.7ex\hbox{E}\kern-.125emX}}

\begin{document}

\Urlmuskip=0mu plus 1mu


\author{\IEEEauthorblockN{
Jiajun Huang,\IEEEauthorrefmark{1}
Sheng Di,\IEEEauthorrefmark{2}
Xiaodong Yu,\IEEEauthorrefmark{3}
Yujia Zhai,\IEEEauthorrefmark{1}
Zhaorui Zhang,\IEEEauthorrefmark{4}
Jinyang Liu,\IEEEauthorrefmark{1}
Xiaoyi Lu,\IEEEauthorrefmark{5}\\
Ken Raffenetti\IEEEauthorrefmark{2}
Hui Zhou\IEEEauthorrefmark{2}
Kai Zhao\IEEEauthorrefmark{6}
Zizhong Chen\IEEEauthorrefmark{1}
Franck Cappello\IEEEauthorrefmark{2}
Yanfei Guo\IEEEauthorrefmark{2}
Rajeev Thakur\IEEEauthorrefmark{2}
}
\IEEEauthorblockA{\IEEEauthorrefmark{1}University of California, Riverside}
\IEEEauthorblockA{\IEEEauthorrefmark{2}Argonne National Laboratory}
\IEEEauthorblockA{\IEEEauthorrefmark{3}Stevens Institute of Technology}
\IEEEauthorblockA{\IEEEauthorrefmark{4}The Hong Kong Polytechnic University}
\IEEEauthorblockA{\IEEEauthorrefmark{5}University of California, Merced}
\IEEEauthorblockA{\IEEEauthorrefmark{6}Florida State University}
\\
\\
\\
}

\title{An Optimized Error-controlled MPI Collective Framework Integrated with Lossy Compression}

\makeatletter
\patchcmd{\@maketitle}
  {\addvspace{0.5\baselineskip}\egroup}
  {\addvspace{-3\baselineskip}\egroup}
  {}
  {}
\makeatother

\maketitle
\thispagestyle{plain}
\pagestyle{plain}


\begin{abstract}

With the ever-increasing computing power of supercomputers and the growing scale of scientific applications, the efficiency of MPI collective communications turns out to be a critical bottleneck in large-scale distributed and parallel processing. The large message size in MPI collectives is particularly concerning because it can significantly degrade the overall parallel performance. To address this issue, prior research simply applies the off-the-shelf fix-rate lossy compressors in the MPI collectives, leading to suboptimal performance, limited generalizability, and unbounded errors. In this paper, we propose a novel solution, called \textit{C-Coll}, which leverages error-bounded lossy compression to significantly reduce the message size, resulting in a substantial reduction in communication cost. The key contributions are three-fold. (1) We develop two general, optimized lossy-compression-based frameworks for both types of MPI collectives (collective data movement as well as collective computation), based on their particular characteristics. Our framework not only reduces communication cost but also preserves data accuracy. (2) We customize SZx, an ultra-fast error-bounded lossy compressor, to meet the specific needs of collective communication. (3) We integrate \textit{C-Coll} into multiple collectives, such as MPI\_Allreduce, MPI\_Scatter, and MPI\_Bcast, and perform a comprehensive evaluation based on real-world scientific datasets. Experiments show that our solution outperforms the original MPI collectives as well as multiple baselines and related efforts by 1.8--2.7$\times$.

\end{abstract}


\maketitle


\section{Introduction}
MPI collectives provide high-performance collective communications in distributed systems, making a significant impact on various research fields such as scientific applications, distributed machine learning, and others~\cite{awan2017s, wang2006pelegant, abadi2016tensorflow, ayala2019impacts, jain2019scaling, abdelmoniem2021efficient}. With the advent of exascale computing and deep learning applications, the demand for large-message MPI collectives has increased. For example, in image classification tasks, VGG19~\cite{simonyan2015very} and ResNet-50~\cite{he2016deep} have 143 million and 25 million parameters, respectively, with communication overheads of 83\% and 72\%~\cite{abdelmoniem2021efficient}. Therefore, optimizing MPI collectives for large messages has become essential~\cite{chunduri2018characterization, Bayatpour2018SALaR, patarasuk2009bandwidth}.

MPI collectives consist of both internode communication and intranode communication, and the former is often the major concern. The overall collective performance is usually limited by the efficiency of internode communications because of limited network bandwidth. Therefore, optimizing internode collective communication is critical to improving the overall performance of MPI collectives. This topic has been a focus of research for decades, with state-of-the-art algorithms achieving notable improvements. However, with the increasing demand for large-message MPI collectives, further optimization remains necessary~\cite{Alm05BlueGene, thakur2005optimization, patarasuk2009bandwidth}. Lossy compression \cite{Di2016SZ,Tao2017SZ,Zhao2020SZauto,Lindstrom2014ZFP} (rather than lossless compression \cite{Deutsch1996gzip, Gaillyzlib,Collet2015zstd}) is a promising solution to mitigate this MPI collective performance issue because of its ability to significantly reduce the message size.

Although lossy compression has been widely used to resolve many other scalability issues in high-performance computing, such as reducing memory footprint \cite{quant-compression}, reducing storage space \cite{nbody-compression,mdz}, and avoiding duplicated computation \cite{pastri}, only a few studies have explored its use in this direction, and all expose certain limitations. To elaborate, Zhou et al. \cite{Zhou2021GPUCOMPRESSION} proposed GPU-compression enhanced point-to-point communication by integrating MPC \cite{yang2015mpc} and 1D fixed-rate ZFP \cite{Lindstrom2014ZFP} into MVAPICH2 \cite{SHM-MVAPICH2}. Their approach, referred to as \textit{CPR-P2P}, simply involved compressing the messages before transmission and decompressing them after reception, leading to significant performance overhead due to the non-negligible time required for compression and decompression. Meanwhile, Zhou et al. \cite{Zhou2022GPUCOMPRESSIONALLTOALL, Zhou2023IPDPS, Zhou2022HiPC} proposed several additional approaches to improve multiple MPI collectives using 1D fixed-rate ZFP \cite{Lindstrom2014ZFP} on GPUs. Their methods, however, focus on fixed-rate compression\footnote{Fixed-rate compression means that the lossy compression would be performed based on a user-specified fixed compression ratio.}, introducing two major limitations: (1) compression errors cannot be bounded, leading to an uncontrolled accuracy, and (2) the compression quality is considerably lower compared to the fixed-accuracy mode\footnote{Fixed-accuracy, also known as error-bounded lossy compression, compresses data based on a user-specified error bound.} in ZFP, as demonstrated by prior research \cite{fraz}.

The aforementioned limitations of compression-enabled MPI collective algorithms motivate us to develop a new efficient MPI collective framework which leverages lossy compression technique to significantly improve the MPI collective performance. However, this brings in three direct technical challenges. \circled{A} Devising a general framework that can effectively hide the communication cost and choose an appropriate timing to call lossy compression is non-trivial.
\circled{B} The data loss nature of the lossy compression brings up a critical concern on the accuracy of collective operations. \circled{C} Existing lossy compressors are not designed for the collective context, leading to suboptimal collective performance because of unnecessary overheads when they are directly applied in MPI collectives \cite{Di2016SZ, Tao2017SZ, Liang2018SZ, Yu2022SZx, Lindstrom2014ZFP}.

Our developed framework is named as compression-facilitated MPI collective framework (\textit{C-Coll}), which can address the aforementioned limitations and challenges. To the best of our knowledge, this is the first-ever framework that provides a general high-performance solution for compression-integrated MPI collectives. Moreover, this is the first accuracy-aware design, which ensures the accelerated collective performance with error-bounded lossy compression does not compromise data quality. To be more specific, our contributions include:
\begin{itemize}
    \item To address challenge \circled{A}, we introduce two efficient frameworks, which can significantly accelerate both types of MPI collectives. Specifically, the first framework notably diminishes the compression overhead in the collective data movement operations (e.g., Scatter, Bcast and All-gather), thereby achieving substantial performance improvement. The second one hides communication inside of compression in collective computation (e.g., Reduce-scatter), which in turn enhances the performance of collective computation. Moreover, these two frameworks can be combined together to speed up more advanced MPI collective operations such as All-reduce. 
    \item To address challenge \circled{B}, we devise several strategies to effectively control the error propagation in the \textit{C-Coll} framework. On the one hand, we carefully select the most suitable error-bounded lossy compressor based on a comprehensive analysis of various state-of-the-art lossy compression methods, considering multiple aspects of their performance on MPI collectives. On the other hand, we carefully regulate the number of compression operations and also develop an efficient method to resolve the imbalance in collective communication introduced by the error-bounded lossy compression.
    \item To address challenge \circled{C}, we customize SZx to cater to the specific needs of collective communication. Specifically, we redesign the compression workflow of SZx and implement a pipelined version of SZx to overlap compression and communication. As a result, we significantly decrease the communication cost of our \textit{C-Coll} framework by up to 4.9$\times$.
    \item To prove the error-bounded nature of the \textit{C-Coll} framework, we perform an in-depth mathematical analysis to derive the limited impact of lossy compression on error propagation. Theoretical analysis demonstrates that the final error of collective data movement framework of \textit{C-Coll} can be well-bounded within the compression error bound set by the user. For collective computation framework, we can derive the final aggregated error is bounded well with a high probability. For example, if there are 100 nodes, the aggregated error is bounded in the range of $[-\frac{20}{3} \widehat{e}, \frac{20}{3} \widehat{e}]$ with a probability of $95.44\%$, where $\widehat{e}$ is the compression error bound.
    \item To demonstrate the generality of our design, we integrate \textit{C-Coll} into multiple collectives, such as MPI\_Allreduce, MPI\_Scatter, and MPI\_Bcast, and evaluate the performance carefully with various scientific datasets generated by real-world applications (such as Seismic Imaging \cite{Kayum2020RTM}, Hurricane Simulation \cite{hurricane2004}, and Climate Simulation \cite{cesm,cesm-code}). Experiments with 128 Xeon Broadwell compute nodes from a supercomputer show that other related efforts or baselines exhibit undesired performance degradation on MPI\_Allreduce because of the significant compression overhead. In comparison, our solution---the lossy-compression-based Allreduce (we call it C-Allreduce)---outperforms the original MPI\_Allreduce by 2.1$\times$. Moreover, our compression-based scatter and Bcast operations (namely C-Scatter and C-Bcast) achieve up to 1.8$\times$ and 2.7$\times$ performance improvements compared with the original MPI\_Scatter and MPI\_Bcast.
    We also use a real-world use-case (image stacking analysis) to validate the practical effectiveness of C-Allreduce, which shows up to 1.5$\times$ performance gain over MPI\_Allreduce, while preserving a high data integrity/accuracy during the collective operations.
\end{itemize}

The rest of the paper is organized as follows: we introduce background and related work in Section \ref{sec:background} and detail our design and optimization in Section \ref{sec-design_and_optimizations}. Evaluation results are presented in Section \ref{exp-setup-sec} followed by conclusion and future work in Section \ref{sec:conclusion}.

\section{Background and Related Work}
\label{sec:background}
In this section, we discuss the background and related work. We first introduce MPI collective communication, followed by a discussion on high-speed lossy compressors and their integration with MPI implementations. The focus of our study is on lossy compression. This emphasis is due to the significantly lower compression ratios observed with lossless methods when applied to scientific datasets~\cite{Di2016SZ,Tao2017SZ}.

\subsection{MPI Collective Communication}
\label{sec:MPI-Collectives}
There are many types of MPI collective operations, which can be divided into two sub-categories --- collective data movement and collective computation according to their communication patterns.
\subsubsection{Collective data movement}
Collective data movement includes gather, allgather, scatter, all-to-all, and so on. The gather operation collects the data from different processes and stores the collected data into the root process. In comparison, allgather stores the collected data to every participated process. As an opposite of gather, the scatter operation divides the data in the root process and sends the split data to all the processes. The all-to-all operation acts as the ``allscatter", which collectively scatters data on each process to each other.

\subsubsection{Collective computation}
Allreduce/reduce are two popular collective computation operations. We use MPI\_SUM as an example to explain the working principle as it is frequently used. The reduce routine will sum up all the data entries from all the processes in the same communicator and store the sum into the root process. The Allreduce does the same thing but keeps a copy of the sum on every process in that communicator. Another widely-used operation is the reduce-scatter, which acts like a combination of the reduce and scatter: the reduced sum is scattered into all the processes.

\subsection{High-speed Lossy Compressors}
Compression developers and scientific researchers have shown great interest in high-speed lossy compression due to its ability to achieve a very high compression ratio. SZ~\cite{Di2016SZ, Tao2017SZ, Liang2018SZ} is an example of a fast, error-bounded lossy compressor that offers performance similar to other compressors, including FPZIP~\cite{Lindstrom2006FPZIP} and SZauto~\cite{Zhao2020SZauto}. Another compressor that is known for its relatively high compression ratios and even faster compression speed than SZ is ZFP~\cite{Lindstrom2014ZFP}. However, neither of these compressors can match the speed of SZx~\cite{Yu2022SZx}, which achieves a compression throughput of 700-900 MB/s on CPUs~\cite{Yu2022SZx}. Additionally, our comprehensive experiments and analysis, as discussed in Section \ref{sec:high-speed-compressor}, demonstrate that SZx is the most suitable lossy compressor for MPI collective operations. As such, we developed our customized compressor for MPI collectives based on SZx, which will be detailed in Section \ref{sec:high-speed-compressor}.

\vspace{-2mm}
\subsection{Lossy Compression-enabled MPI Implementations}
Researchers have shown interest in using lossy compression to improve MPI communication performance for years. Zhou et al. proposed GPU-compression enhanced point-to-point communication \cite{Zhou2021GPUCOMPRESSION}, and several optimized MPI collective operations \cite{Zhou2022GPUCOMPRESSIONALLTOALL, Zhou2023IPDPS, Zhou2022HiPC} using 1D fixed-rate ZFP~\cite{Lindstrom2014ZFP} on GPUs, but their solutions are either showcasing limited overlapping between compression and communication or subject to the fixed-rate mode compression, leading to the substandard compression quality, as will be demonstrated later in Section \ref{sec-image-stacking}. Conversely, our general framework can optimize the collective performance of all MPI collectives while maintaining controlled errors.

\section{C-Coll Design and Optimization}
\label{sec-design_and_optimizations}

Figure \ref{fig:architecture} presents the overall design architecture of \textit{C-Coll}. We highlight the newly designed modules as green boxes. The primary contributions lie in the performance optimization layer and the middleware layer. We carefully characterize the performance of multiple state-of-the-art error-bounded lossy compressors in the context of MPI collectives and select the best-qualified compressor -- SZx, which will be detailed in Section \ref{sec:high-speed-compressor}. We also propose a series of performance optimization strategies specifically for both of the two collective types (data movement and collective computation), as indicated in the figure. The corresponding details will be discussed in Sections \ref{sec-data-movement-framework}, \ref{sec-collective-computation}, and \ref{sec:Customize-SZx-reduce}. 

\begin{figure}[ht]
    \centering
    {\includegraphics[width=1\linewidth]{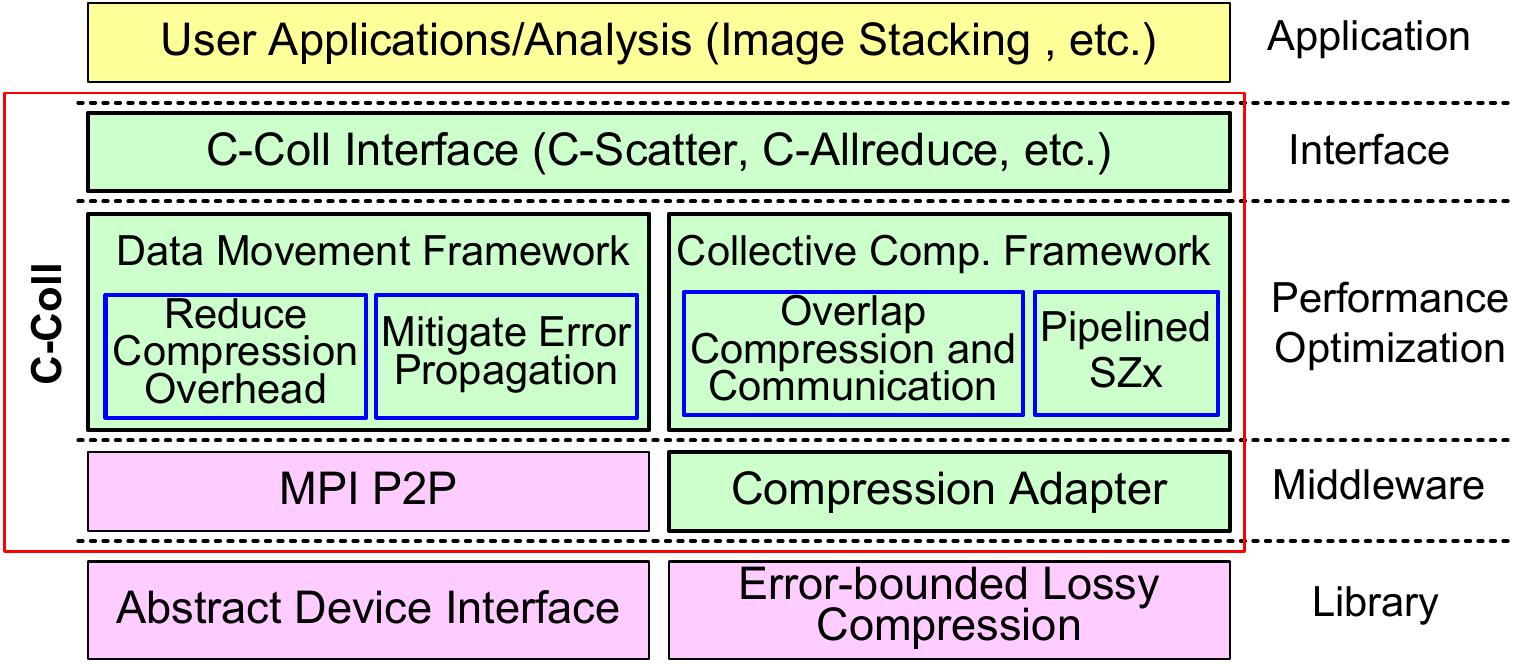}}
    \caption{Design architecture (yellow box: applications; green box: new contributed modules; purple box: third-party).} 
    \label{fig:architecture}
\end{figure}


\subsection{Two Proposed Novel Frameworks for Compression-enhanced Collectives}
To integrate lossy compression into MPI collective communications, at least two important aspects must be considered: performance and accuracy. In general, MPI collective operations can be divided into two groups: collective data movements, and collective computation. Instead of directly using the CPR-P2P method, we propose two frameworks to implement collective communications for each group, which can maximize the collective operation performance.

\subsubsection{\textbf{Collective data movement framework}}
\label{sec-data-movement-framework}

In this subsection, we detail our strategies for addressing the issues of communication imbalance and compression overhead in our optimized collective data movement framework.
For collective data movement operations, each process in the same communicator needs to communicate with each other to exchange data. If we directly use the CPR-P2P method, the sender needs to compress the data every time before sending it, and the receiver is required to decompress the data upon the data arrival. Most of the compression and decompression overheads in CPR-P2P, actually, can be avoided by carefully setting the timing of compression operations, as the original data have not been modified during the intensive communications. Besides, the CPR-P2P can cause unbalanced communications in that input data on different processes have various compressed data sizes. Such unbalanced communication will slow down the overall collective performance, resulting in a sub-optimal performance. With our framework, we can balance the communications with a fixed pipeline size as the compressed data sizes are decided at the beginning of the intensive communications. In the following text, we illustrate our idea with two examples: a ring-based allgather algorithm and a binomial tree broadcast algorithm. The same philosophy can be easily extended to other collective algorithms in this category.
\begin{figure}[ht]
    \centering
    {\includegraphics[width=0.8\linewidth]{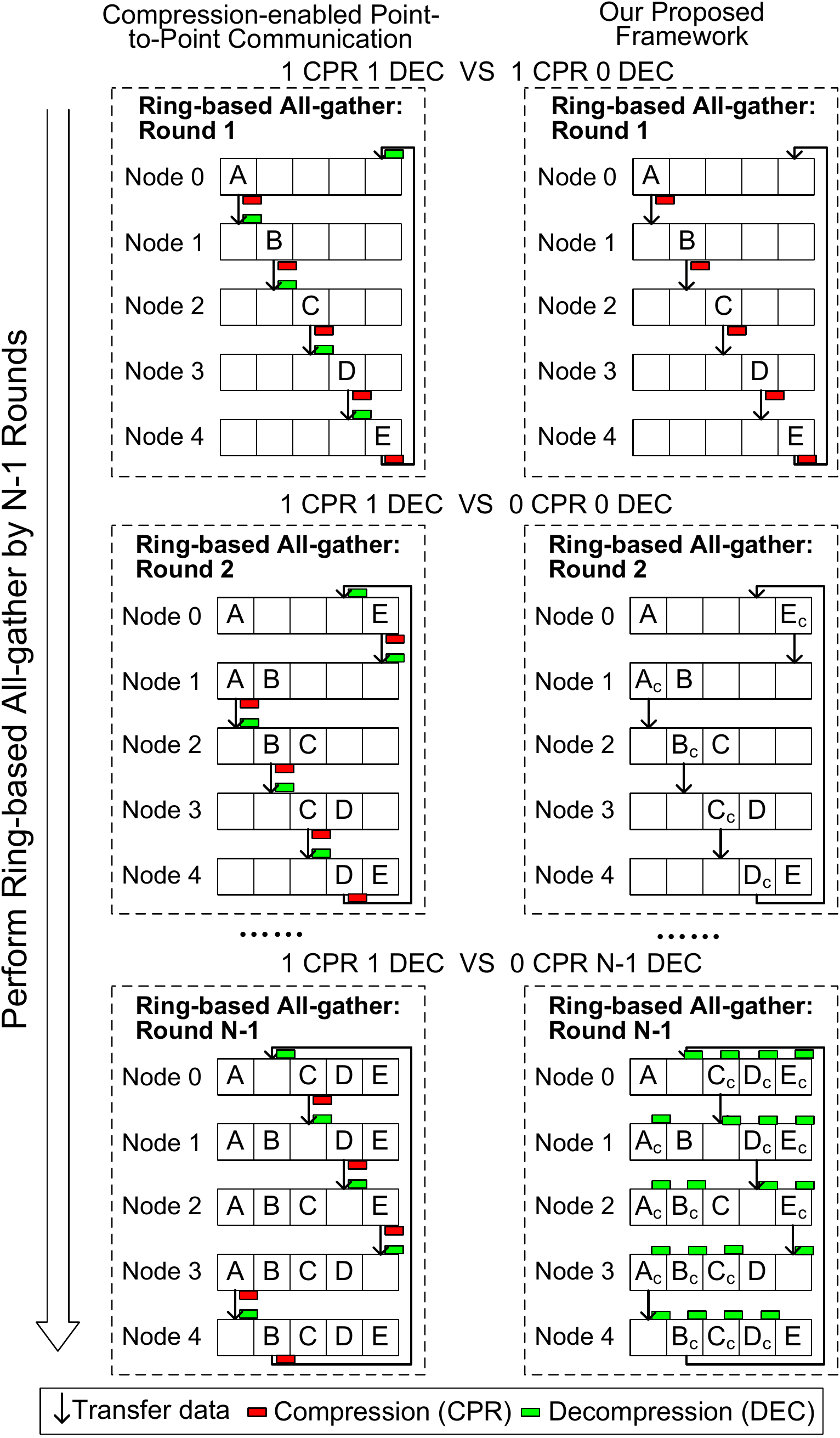}}
    \caption{High-level design of our collective data movement framework in the ring-based allgather algorithm to mitigate compression error propagation. $A$ means the original data and $A_c$ means the compressed data. This rule applies to other data chunks as well. This algorithm completes in $N$$-$1 rounds, where $N$ is the number of processes.} 
    \label{fig-allgather-design}
    \vspace{-2mm}
\end{figure}

Figure \ref{fig-allgather-design} shows the high-level comparison of our proposed framework versus the CPR-P2P in the ring-based allgather algorithm. In this algorithm, $N$$-$1 rounds are required to get the gathered results on every process, where $N$ is the number of processes in the communicator. In order to use lossy compression to reduce communication cost in the ring-based algorithm, the straight-forward idea is performing compression and decompression at each round. Instead, our design does not decompress the received data until the last round, thus significantly decreasing the compression overhead from ($N$$-$1)$\cdot$$T_{chunk}$ to $T_{chunk}$, where $T_{chunk}$ is the compression cost of one chunk. When $N$ is large, our novel framework could have nearly $N$$\times$ better performance compared with CPR-P2P in terms of compression. Note that the decompression cost remains the same. 

Figure \ref{fig-broadcast-design} presents a similar comparison in the binomial tree broadcast algorithm. There are $log_2 N$ rounds before the completion, where $N$ refers to the total process count. We notice that our proposed framework could reduce the compression and decompression costs from $log_2$$N$$\cdot$($T_{comp}$+$T_{decom}$) to $T_{comp}$+$T_{decom}$ and the performance improvement is $log_2 N$$\times$, where $T_{comp}$ and $T_{decom}$ are the compression time and decompression time of the data at the root process, respectively.

Besides the compression cost, the CPR-P2P method can lead to an undesirable error propagation issue during collective sends and receives, as the same data chunk undergoes multiple rounds of compression and decompression. Our framework also solves this issue by compressing the same data chunk for only one time. Similar to the analysis of compression overhead, for the absolute error bounded compression, our proposed framework can decrease the worst case accuracy loss by ($N$$-$1)$\times$ and ($log_2 N$)$\times$ in the ring-based allgather and binomial tree broadcast algorithm, respectively.

\subsubsection{\textbf{Collective computation framework}}
\label{sec-collective-computation}
For collective computation routines, the data entries from all processes in the same communicator need to collectively compute with each other. Unlike in the case of collective data movement, the data transferred in this communication pattern can be updated. As a result, the previous framework cannot be utilized here thus we need to propose a new framework. Despite the updated transferred data precluding us from diminishing compression, we find an opportunity to hide communication inside the compression and decompression.
To clearly elaborate our proposed design, we use the ring-based reduce\_scatter algorithm as an instance. Note that this framework can be easily extended to other collective computation operations. 

\begin{figure}[ht]
    \centering
    {\includegraphics[width=0.8\linewidth]{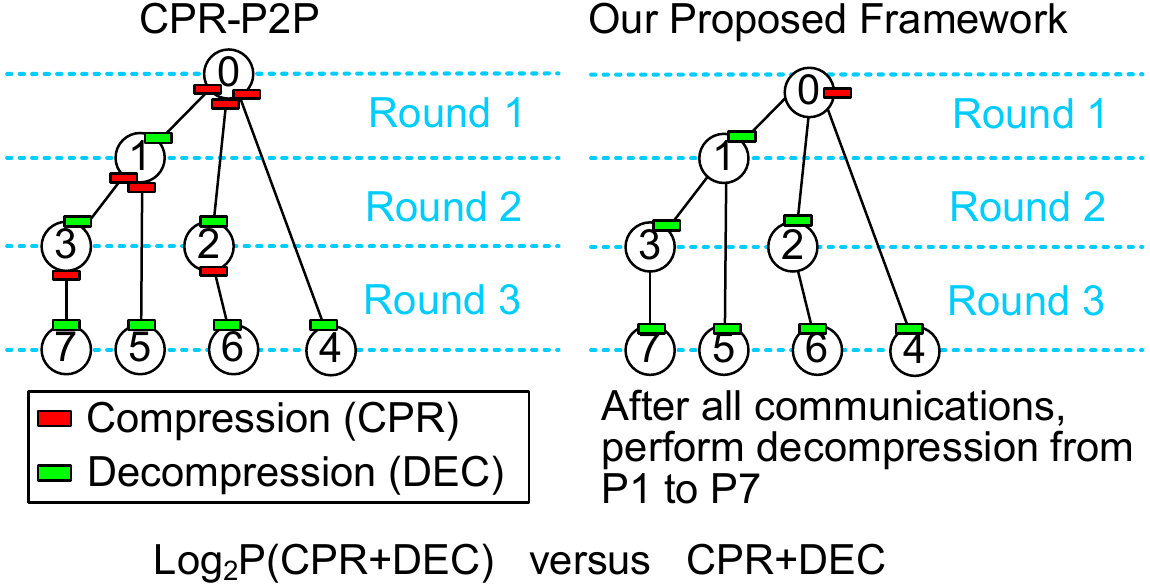}}
    \caption{High-level design of our collective data movement
    framework in the binomial tree broadcast algorithm. It completes in $log_2{N}$ rounds, where $N$ is the number of processes.} 
    \label{fig-broadcast-design}
\end{figure}

Our proposed framework for collective computation is depicted in Figure \ref{fig-reduce-scatter-design}. It employs a ring-based reduce\_scatter algorithm, where each process is required to exchange message chunks with its neighboring processes. For large datasets, these chunk sizes become substantial as they are determined by dividing the size of the input data by the number of processes. In the initial CPR-P2P model, compression and decompression occur before and after any communication, respectively. Consequently, a single round incurs three types of overhead: compression/decompression for one message chunk, send/receive operations for the compressed message chunk, and reduction operation for one message chunk. Typically, the compression-related overhead is more significant than the send/receive overhead, as compressed data sizes are considerably smaller than their original counterparts. In our redesigned approach, we significantly mitigate the send/receive overhead by actively pulling communication progress within the compression and decompression phases. This substantially reduces the overall communication time.

\begin{figure}[ht]
    \centering
    {\includegraphics[width=0.8\linewidth]{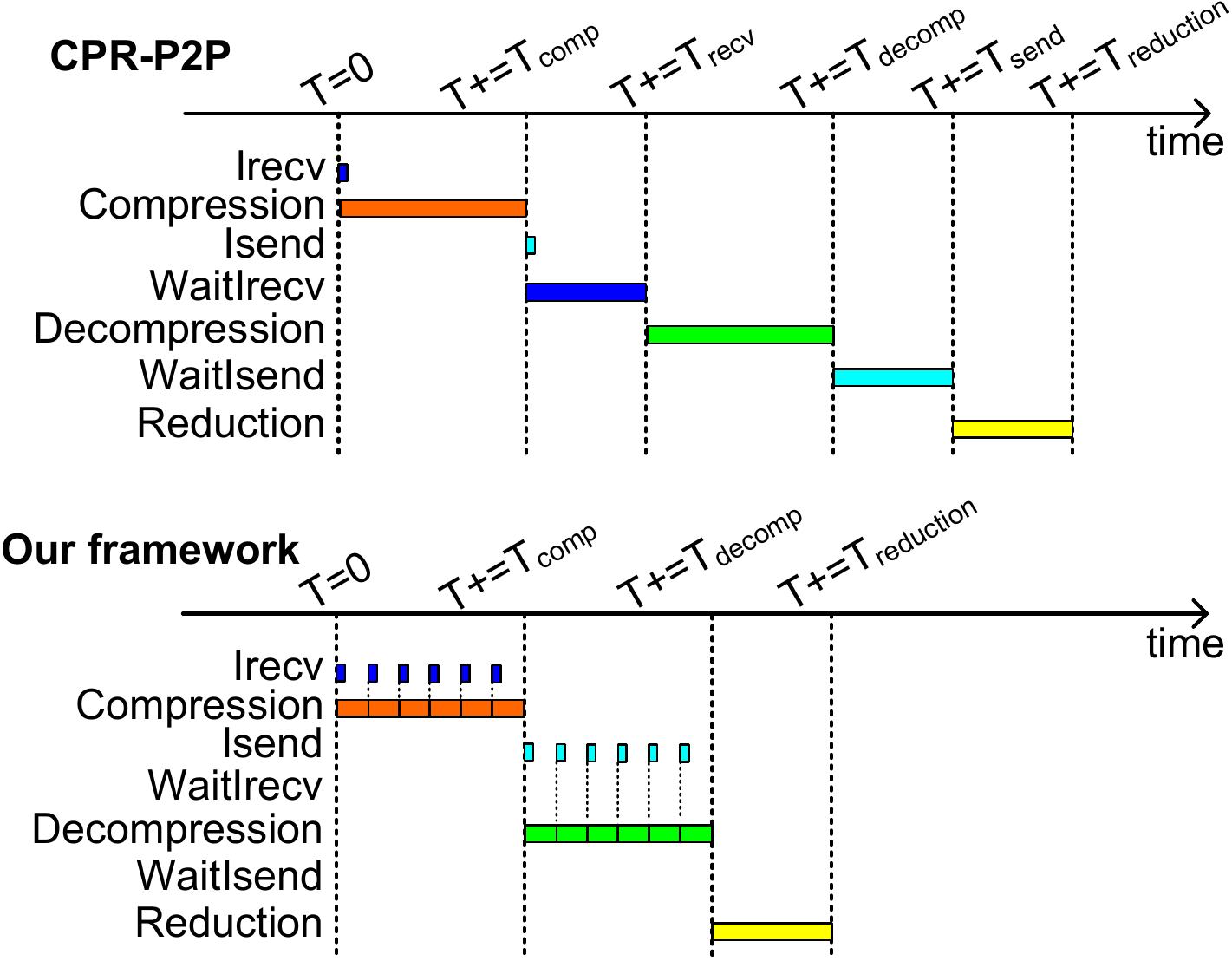}}
    \caption{High-level design of our collective computation framework in the ring-based reduce-scatter algorithm.} 
    \label{fig-reduce-scatter-design}
\end{figure}


\subsection{Theoretical Analysis of Error Propagation in C-Coll}
In this section, we prove the error-bounding nature of our \textit{C-Coll} framework mathematically. In the following analysis, we assume the lossy compression error $e$ for data $x$ follows a normal distribution, without loss of generality. Specifically, the normal distribution is represented as $e \sim N(\mu, \sigma^{2})$ within the range of $[x- \widehat{e}, x+\widehat{e}]$, where $\widehat{e}$ is the compression error bound. This is exemplified in Figure \ref{fig:normal-distri}, in which we compress climate, weather, seismic wave datasets by SZ3 and ZFP. It clearly shows the normal distribution curve generated by Maximum Likelihood Estimation (MLE) fits the measured compression error values very well for different application datasets.

\vspace{-3mm}
\begin{figure}[ht]
    \centering
        \subfloat[SZ3(Climate)]        {\includegraphics[width=0.275\linewidth]{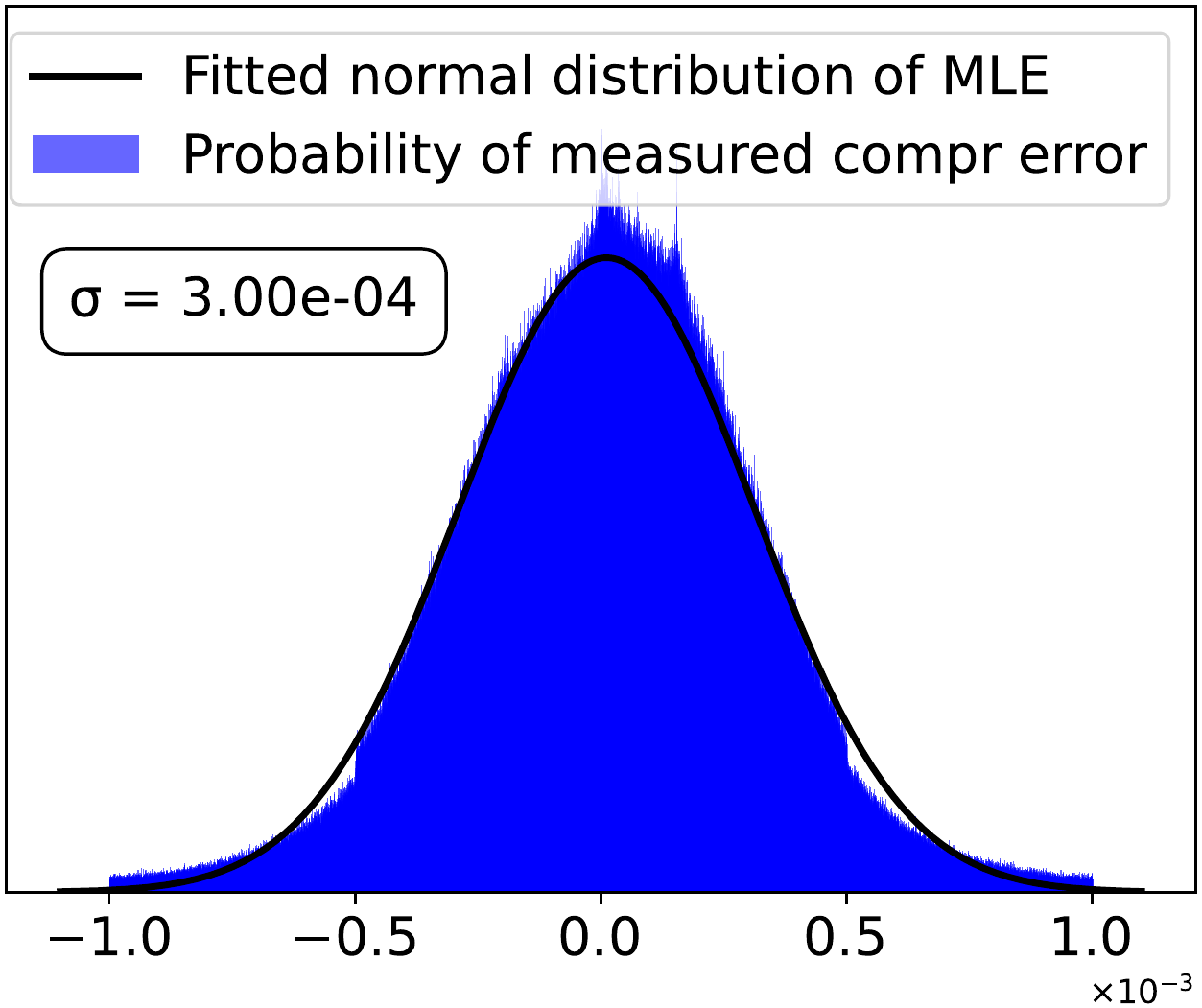}}
        \subfloat[SZ3(Weather)]        {\includegraphics[width=0.35\linewidth]{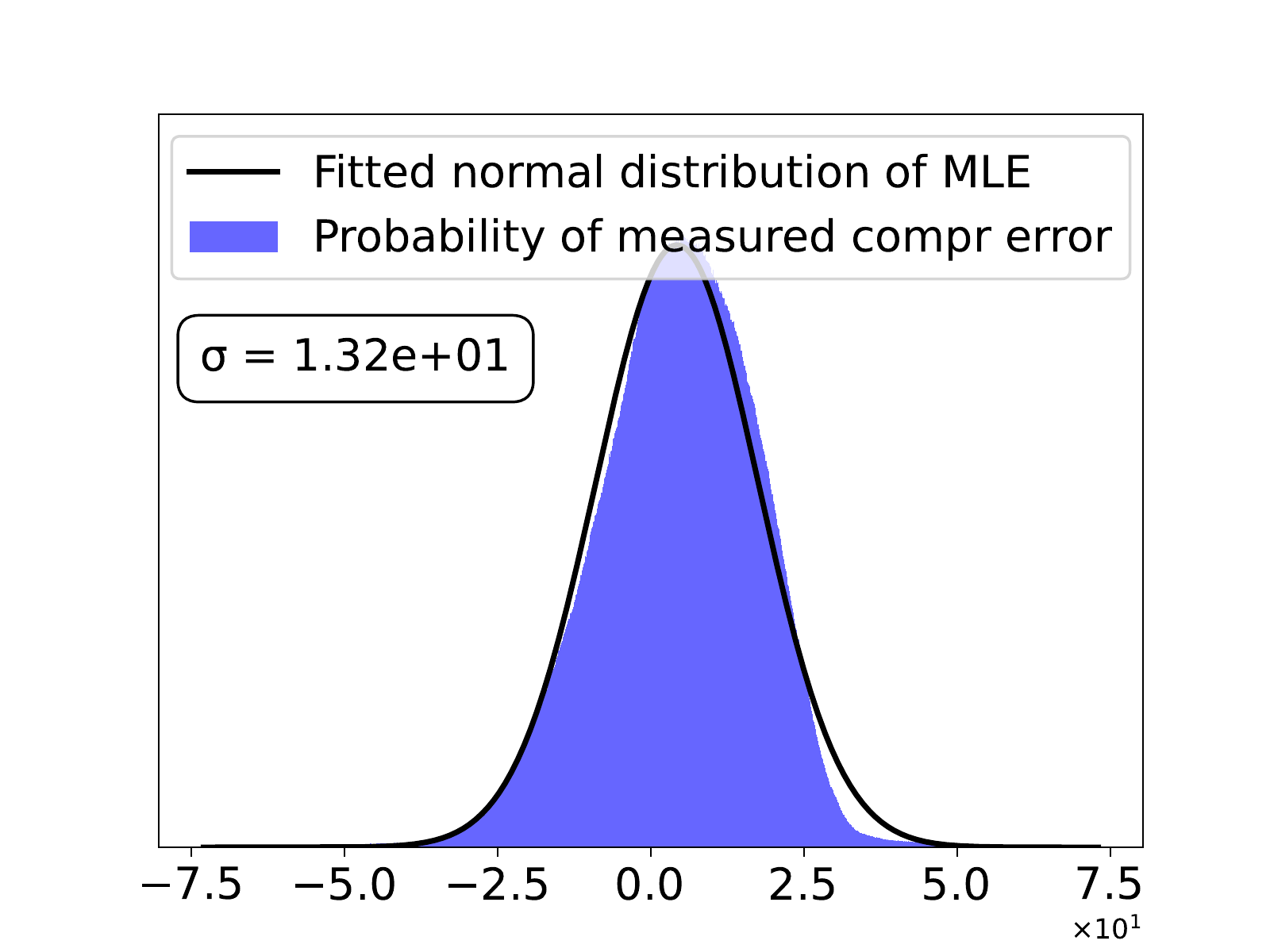}}
        \subfloat[SZ3(Seismic Wave)]   {\includegraphics[width=0.35\linewidth]{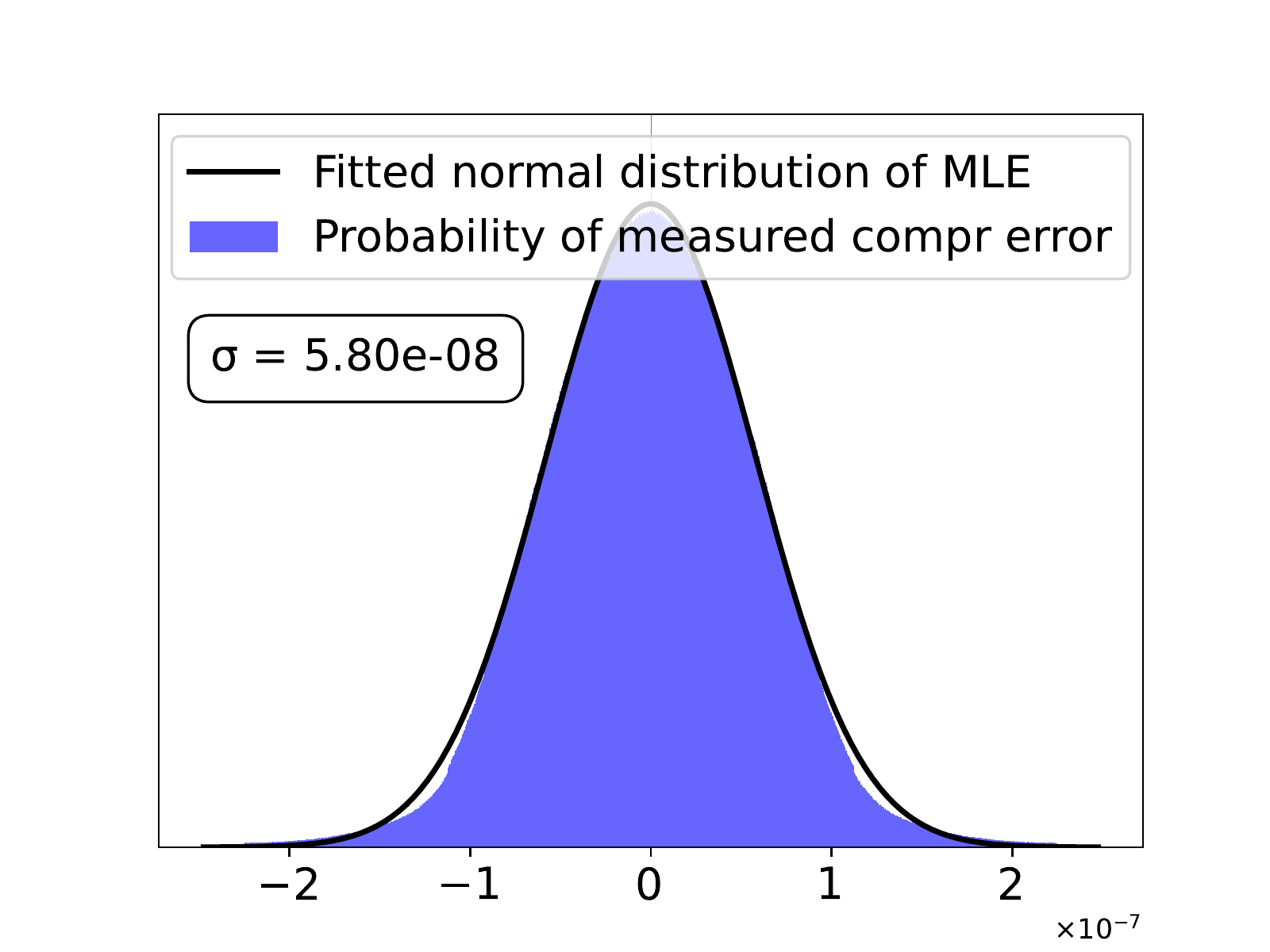}}
        \hspace{3mm}
        \subfloat[ZFP(Climate)]{\includegraphics[width=0.275\linewidth]{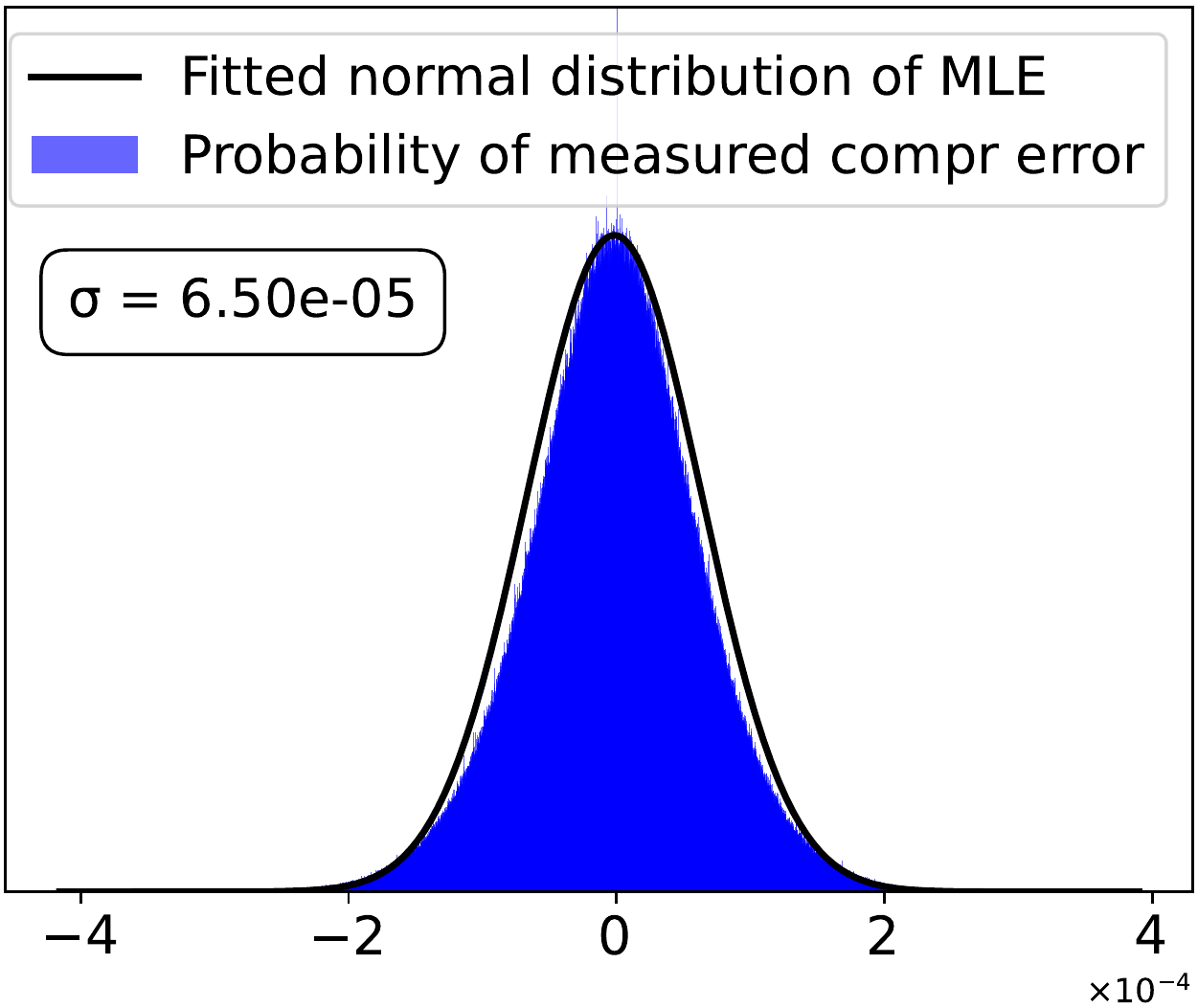}}
        \subfloat[ZFP(Weather)]
        {\includegraphics[width=0.35\linewidth]{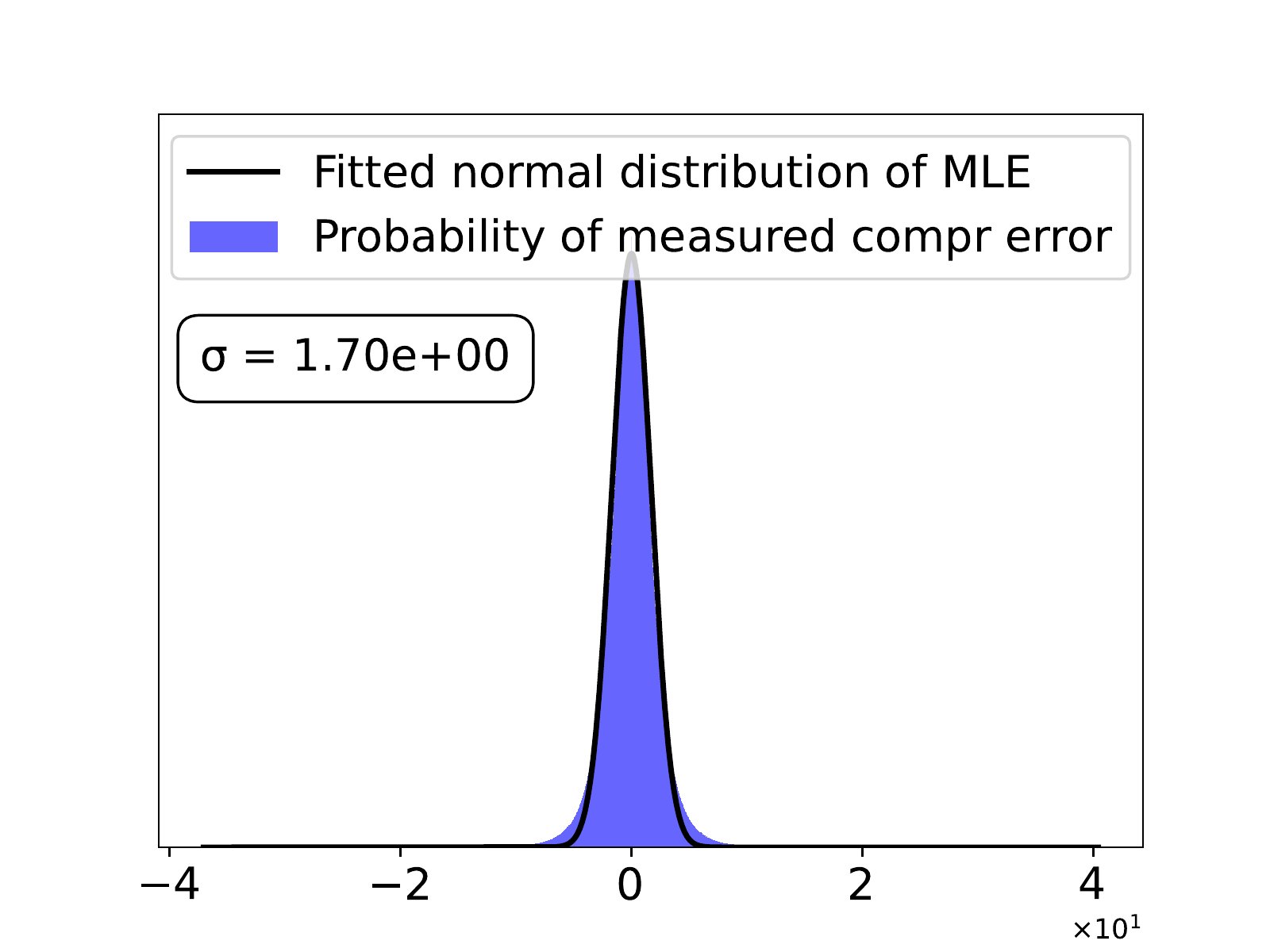}}
        \subfloat[ZFP(Seismic Wave)]
        {\includegraphics[width=0.35\linewidth]{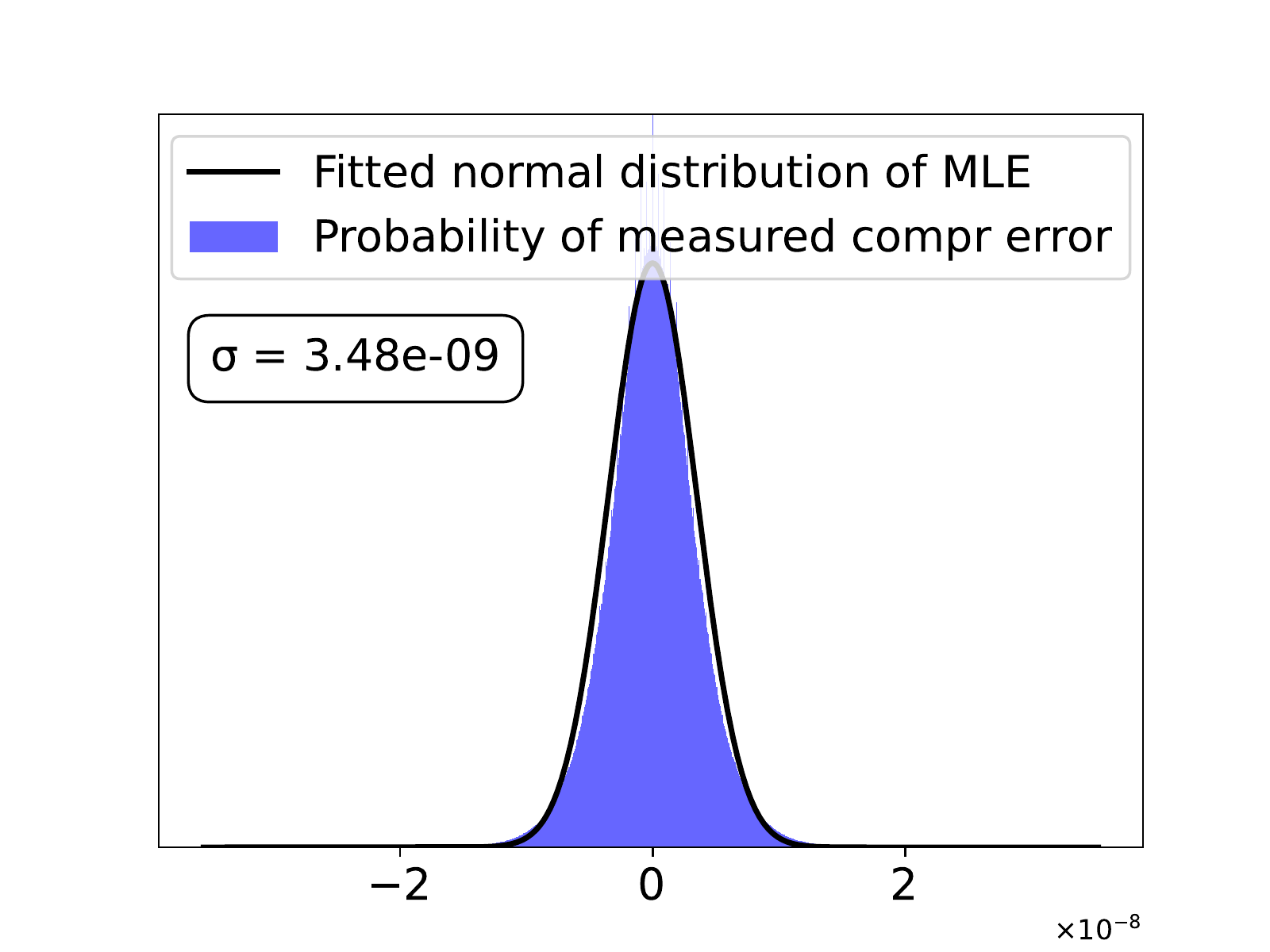}}
        \caption{Exemplifying the normal distribution property of compression errors.}
        \label{fig:normal-distri}
\end{figure}

For the collective communication primitive in MPI, the data will be aggregated gradually from each node during the communication process, which is shown in Fig. \ref{fig-allgather-design}, Fig. \ref{fig-broadcast-design} and Fig. \ref{fig-reduce-scatter-design}. With lossy compression integrated, the compression error will also be aggregated in the data aggregation stage. In the collective data movement framework of \textit{C-Coll}, each data chunk is compressed only once. Thus, the final error for each data point is within $\widehat{e}$. Unlike the data movement framework, the compression error is aggregated in the collective computation framework of \textit{C-Coll}, and the aggregation function often involves the \textit{Sum, Average, Max, Min} operations for the floating point data. We further illustrate the error propagation for these operations. For the collective computation framework, we assume there are $n$ data that are collected from $n$ nodes, and the compression error for each of them is $e_{i}$, which follows the normal distribution $e_{i} \sim N(\mu_{i}, \sigma_{i}^{2})$.

\begin{theorem}
Based on the above analysis, the final aggregated error for \textit{Sum} operation falls into the interval $[-2\sqrt{n}\sigma, 2\sqrt{n}\sigma]$ with the probability of $95.44\%$, where $n$ is the number of computing nodes in MPI and $\sigma$ is the variance of the error bound of the lossy compressor.
\end{theorem}

\begin{proof}
The linear combination (e.g., \textit{Sum} in MPI) of normally distributed random variables also follows a normal distribution that is shown in the formula (\ref{norm_distribution}), where $a_{i}$ denotes various constants for $n$ data.

\begin{equation}
\sum\nolimits_{i=0}^{n} a_{i}e_{i} \sim N (\sum\nolimits_{i=0}^{n} a_{i}\mu_{i}, \sum\nolimits_{i=0}^{n} a_{i}^{2} \sigma_{i}^{2})
\label{norm_distribution}
\end{equation}

For the \textit{Sum} operation in the collective computation framework, the compression error will be involved gradually with the aggregation chain, which is shown in the formula (\ref{sum_aggregation}):

\begin{equation}
\begin{aligned}
x_{sum} & = \big(\big(\big((x_{1} + e_{1})+x_{2}\big)+e_{2}\big)+...+e_{n}\big) \\
& = (x_{1}+e_{1})+(x_{2}+e_{2})+...+(x_{n}+e_{n})
\end{aligned}
\label{sum_aggregation}
\end{equation}

Where $x_{i}$ indicates the data collected from node $i$ and $x_{sum}$ represents the final aggregated data. We further demonstrate that the compression error $e_2$ remains to follow the normal distribution after compression through Figure \ref{fig:e2-normal-distri}, which illustrates the probability of the measured compression error ($e_2$) with the fitted MLE curve. The same property also holds for subsequent compression errors in the aggregation chain such as $e_3$, $e_4$, and so on.

\vspace{-3mm}
\begin{figure}[ht]
    \centering
        \subfloat[SZ3($e_2$)]        {\includegraphics[width=0.35\linewidth]{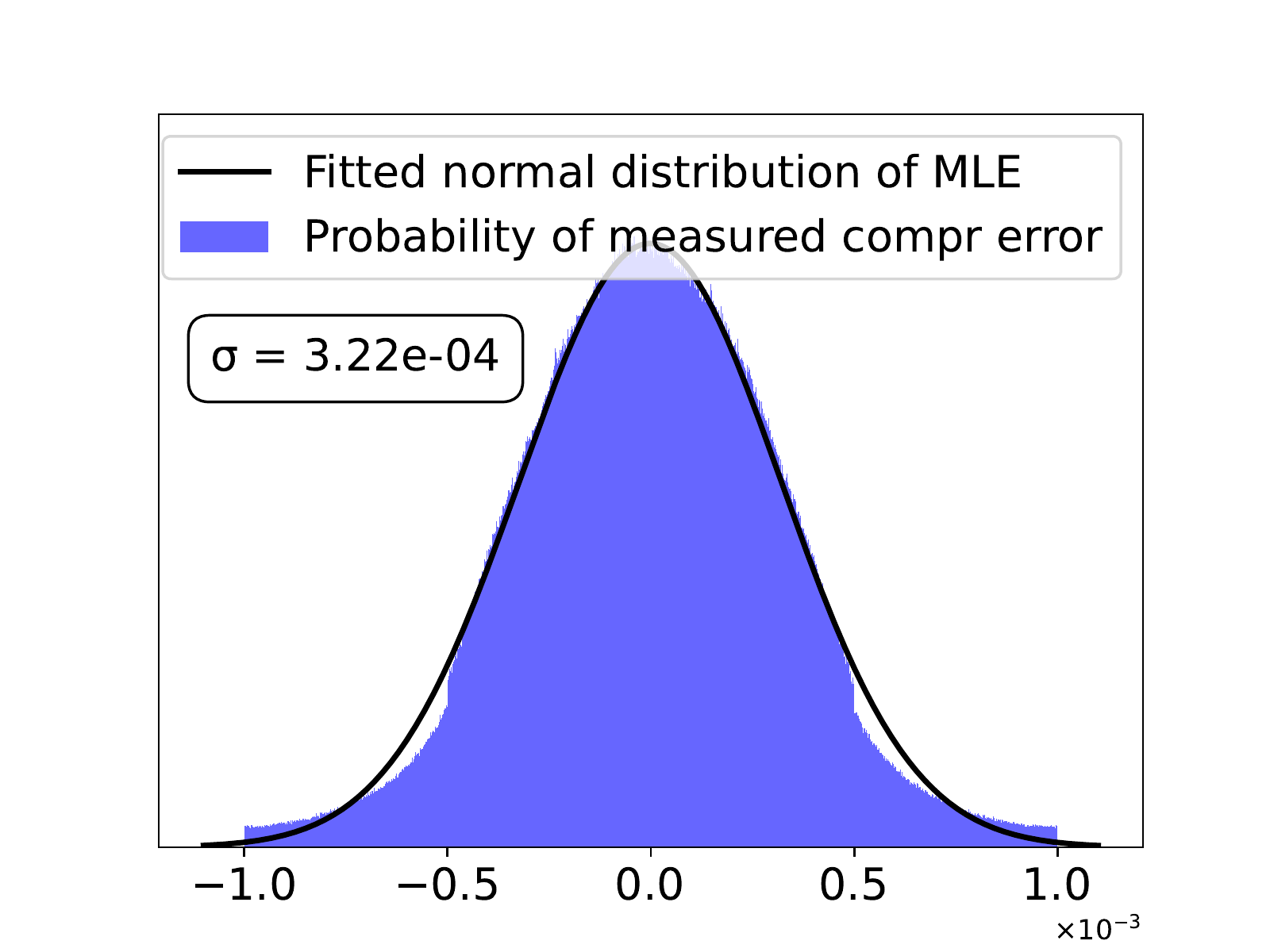}}
        \hspace{3mm}
        \subfloat[ZFP($e_2$)]
        {\includegraphics[width=0.35\linewidth]{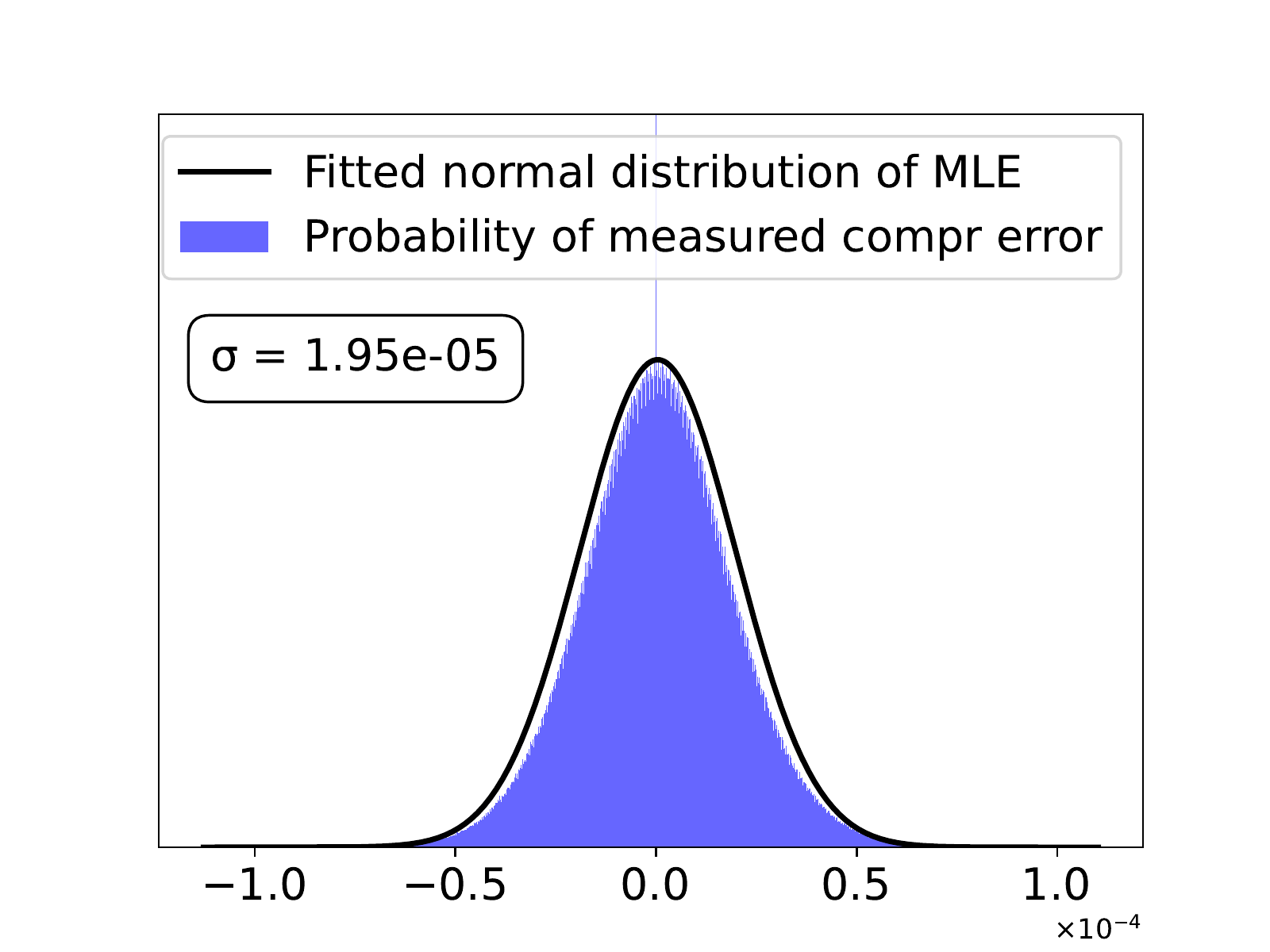}}
        \caption{Exemplifying the normal distribution property of compression error $e_2$.}
        \label{fig:e2-normal-distri}
\end{figure}

From the formula (\ref{sum_aggregation}), we can calculate the aggregated compression error as $\widetilde{e}_{sum} = \sum_{i=0}^{n} e_{i}$. The aggregated error $\widetilde{e}_{sum}$ follows the normal distribution as formula (\ref{error_distribution}):

\begin{equation}
\widetilde{e}_{sum} \sim N(\sum\nolimits_{i=0}^{n} \mu_{i}, \sum\nolimits_{i=0}^{n} \sigma_{i}^{2})
\label{error_distribution}
\end{equation}

The formula (\ref{error_distribution}) indicates that the variance $\sigma^{2}$ of the aggregated error $\widetilde{e}_{sum}$ will be restricted well. When we utilize the same compression error bound across various nodes, the aggregated error conforms to a normal distribution represented as $\widetilde{e}_{sum} \sim N(0, n\sigma^{2})$. Therefore, the final aggregated error falls within the interval $[-2\sqrt{n} \sigma, 2\sqrt{n} \sigma]$ with the probability of $95.44\%$ according to the properties of normal distribution. 

\end{proof}

Since the compression error is bounded by the error bound $\widehat{e}_{i}$ and the error follows the norm distribution, we can assume that $\widehat{e}_{i} \approx 3 \sigma_{i}$ ($\widehat{e}_{i}$ bounded to $3 \sigma_{i}$ with probability of $99.74\%$).

\begin{corollary}
Based on the above assumption, the final aggregated error falls within the interval $[-\frac{2}{3}\sqrt{n} \widehat{e}, \frac{2}{3}\sqrt{n} \widehat{e}]$ with the probability of $95.44\%$. For example, if there are $100$ nodes, the final aggregated error will be bounded within the range $[-\frac{20}{3} \widehat{e}, \frac{20}{3} \widehat{e}]$ with a probability of $95.44\%$.

\end{corollary}

\begin{corollary}
Being similar to the \textit{Sum} operation, the final aggregated error for the \textit{Average} operation in collective computation framework follows the normal distribution $\widehat{e}_{avg} \sim N(0, \frac{\sigma^{2}}{n})$. The final aggregated error will be reduced extremely compared to the original error by $n$ times. 
\end{corollary}

\begin{theorem}
For the \textit{Max, Min} operations, the final error follows the normal distribution $ \widetilde{e}_{max, min} \sim N (0, (2-\frac{n+2}{2^{n}})\sigma^{2})$.
\end{theorem}

\begin{proof}
Since we need to compare the data from the neighbored node gradually, there is a $\frac{1}{2}$ probability that we can choose the non-compressed data. Otherwise, the selected data will contain an error within the error bound $\widehat{e}$. Therefore, the variance of the final aggregated error can be calculated as following formula (\ref{max_min}):
\begin{equation}
\vspace{-3mm}
\frac{1}{2^{n}}n\sigma^{2} + \frac{1}{2^{n-1}}(n-1)\sigma^{2} +...+\frac{1}{2}\sigma^{2} = (2-\frac{n+2}{2^{n}}) \sigma^{2}
\label{max_min}
\end{equation}
\end{proof}

\subsection{Identify Best-qualified High-speed Error-bounded Lossy Compressor}
\label{sec:high-speed-compressor}
In this section, we compare various lossy compressors and select the most suitable one for MPI collectives. As shown in the previous analysis, in addition to controlling the data distortion by error bounds, the compression throughput and compression ratio are two critical metrics to consider.
According to the prior literature \cite{Di2016SZ,Zhao2020SZauto,sz3,Liang2018SZ,Yu2022SZx}, ZFP and SZx exhibit much higher compression speed than other compressors, including SZ2 \cite{Liang2018SZ}, SZ3 \cite{sz3}, FPZIP \cite{Lindstrom2006FPZIP}, and Auto-SZ \cite{Zhao2020SZauto}. Therefore, we focus on ZFP and SZx in particular and select the best-qualified one for compression-enabled collective communication.

In ZFP, users can set a fixed compression rate in the fixed-rate (FXR) mode or a fixed absolute error bound in the fixed-accuracy (ABS) mode. The two modes have their particular pros and cons. The advantage of the FXR is that it brings up substantial convenience in some use cases that require knowing the compressed data size in advance. In comparison with FXR mode, ZFP's ABS mode features a much higher compression quality on scientific datasets (i.e., higher compression ratio with the same reconstructed data quality), as verified in \cite{fraz}. We briefly analyze the key reason why ZFP's ABS mode has higher compression quality than its FXR mode in the following text. ZFP divides the dataset into small data blocks during the compression. Unlike ZFP(ABS) which allows variable compressed size for each block, ZFP(FXR) strictly forces each block to have the same compression ratio, which may inevitably degrade the quality of reconstructed data in turn. Moreover, the FXR mode cannot control the error bound, which may cause fairly high compression errors on some data points unexpectedly. Compared with ZFP, SZx is an ultra-fast error-bounded compressor driven by the fixed-accuracy mode. 

\begin{table} [ht]
\centering
\caption{OVERALL COMPRESSION/DECOMPRESSION THROUGHPUT (MB/S)}
\label{tab:compression-speed}
\resizebox{0.8\columnwidth}{!}{%
\begin{tabular}{|c|c|cc|cc|cc|}
\hline
\textbf{Datasets} & \textbf{} & \multicolumn{2}{c|}{\textbf{RTM}} & \multicolumn{2}{c|}{\textbf{Hurricane}} & \multicolumn{2}{c|}{\textbf{CESM-ATM}} \\ \hline\hline
\multirow{4}{*}{\textbf{SZx}}  & \textbf{ABS} & \multicolumn{1}{c|}{Com} & Decom & \multicolumn{1}{c|}{Com} & Decom & \multicolumn{1}{c|}{Com} & Decom \\ \cline{2-8}
 & 1E-2 & \multicolumn{1}{c|}{1742} & 3309 & \multicolumn{1}{c|}{1687} & 3640 & \multicolumn{1}{c|}{666} & 1251 \\ \cline{2-8} 
 & 1E-3 & \multicolumn{1}{c|}{1479} & 2723 & \multicolumn{1}{c|}{895} & 1659 & \multicolumn{1}{c|}{533} & 918 \\ \cline{2-8} 
 & 1E-4 & \multicolumn{1}{c|}{1288} & 2215 & \multicolumn{1}{c|}{644} & 1168 & \multicolumn{1}{c|}{526} & 822 \\ \hline\hline
  
\multirow{4}{*}{\textbf{ZFP(ABS)}} & \textbf{ABS} & \multicolumn{1}{c|}{Com} & Decom & \multicolumn{1}{c|}{Com} & Decom & \multicolumn{1}{c|}{Com} & Decom \\ \cline{2-8}
& 1E-2 & \multicolumn{1}{c|}{1383} & 1444 & \multicolumn{1}{c|}{492} & 634 & \multicolumn{1}{c|}{240} & 273 \\ \cline{2-8} 
 & 1E-3 & \multicolumn{1}{c|}{1082} & 1141 & \multicolumn{1}{c|}{307} & 397 & \multicolumn{1}{c|}{170} & 191 \\ \cline{2-8} 
 & 1E-4 & \multicolumn{1}{c|}{783} & 811 & \multicolumn{1}{c|}{170} & 209 & \multicolumn{1}{c|}{128} & 136 \\ \hline\hline
\multirow{4}{*}{\textbf{ZFP(FXR)}}& \textbf{FXR} & \multicolumn{1}{c|}{Com} & Decom & \multicolumn{1}{c|}{Com} & Decom & \multicolumn{1}{c|}{Com} & Decom \\ \cline{2-8} 
& 4 & \multicolumn{1}{c|}{610} & 601 & \multicolumn{1}{c|}{251} & 319 & \multicolumn{1}{c|}{335} & 397 \\ \cline{2-8} 
 & 8 & \multicolumn{1}{c|}{438} & 413 & \multicolumn{1}{c|}{129} & 142 & \multicolumn{1}{c|}{200} & 203 \\ \cline{2-8} 
 & 16 & \multicolumn{1}{c|}{324} & 311 & \multicolumn{1}{c|}{82} & 81 & \multicolumn{1}{c|}{119} & 112 \\ \hline
\end{tabular}%
}
\end{table}

\begin{table} [ht]
\centering
\caption{COMPRESSION RATIOS (ORIGINAL DATA SIZE / COMPRESSED DATA SIZE)}
\label{tab:COMPRESSION-RATIOS}
\resizebox{1\columnwidth}{!}{%
\begin{tabular}{|c|c|ccc|ccc|ccc|}
\hline
\textbf{Datasets} & \textbf{} & \multicolumn{3}{c|}{\textbf{RTM}} & \multicolumn{3}{c|}{\textbf{Hurricane}} & \multicolumn{3}{c|}{\textbf{CESM-ATM}} \\ \hline\hline
\multirow{4}{*}{\textbf{SZx}}  & \textbf{ABS} & \multicolumn{1}{c|}{min} & \multicolumn{1}{c|}{avg} & max & \multicolumn{1}{c|}{min} & \multicolumn{1}{c|}{avg} & max & \multicolumn{1}{c|}{min} & \multicolumn{1}{c|}{avg} & max \\ \cline{2-11}
 & 1E-2 & \multicolumn{1}{c|}{88} & \multicolumn{1}{c|}{116.3} & 124.1 & \multicolumn{1}{c|}{121} & \multicolumn{1}{c|}{123.1} & 124.1 & \multicolumn{1}{c|}{4.9} & \multicolumn{1}{c|}{8.5} & 22.8 \\ \cline{2-11} 
 & 1E-3 & \multicolumn{1}{c|}{26.1} & \multicolumn{1}{c|}{49.4} & 115.1 & \multicolumn{1}{c|}{14.9} & \multicolumn{1}{c|}{17.4} & 29.1 & \multicolumn{1}{c|}{3.3} & \multicolumn{1}{c|}{5.1} & 13.1 \\ \cline{2-11} 
 & 1E-4 & \multicolumn{1}{c|}{9.5} & \multicolumn{1}{c|}{30.4} & 111.3 & \multicolumn{1}{c|}{6.9} & \multicolumn{1}{c|}{7.2} & 8.8 & \multicolumn{1}{c|}{2.4} & \multicolumn{1}{c|}{3.4} & 8 \\ \hline\hline
\multirow{4}{*}{\textbf{ZFP(ABS)}}  & \textbf{ABS} & \multicolumn{1}{c|}{min} & \multicolumn{1}{c|}{avg} & max & \multicolumn{1}{c|}{min} & \multicolumn{1}{c|}{avg} & max & \multicolumn{1}{c|}{min} & \multicolumn{1}{c|}{avg} & max \\ \cline{2-11} 
& 1E-2 & \multicolumn{1}{c|}{67.7} & \multicolumn{1}{c|}{87.6} & 125.2 & \multicolumn{1}{c|}{18.1} & \multicolumn{1}{c|}{18.3} & 18.6 & \multicolumn{1}{c|}{4.7} & \multicolumn{1}{c|}{8.1} & 23.1 \\ \cline{2-11} 
 & 1E-3 & \multicolumn{1}{c|}{30.6} & \multicolumn{1}{c|}{58.4} & 123.2 & \multicolumn{1}{c|}{10.5} & \multicolumn{1}{c|}{10.7} & 11 & \multicolumn{1}{c|}{3.4} & \multicolumn{1}{c|}{5.6} & 15 \\ \cline{2-11} 
 & 1E-4 & \multicolumn{1}{c|}{13.4} & \multicolumn{1}{c|}{38} & 120.1 & \multicolumn{1}{c|}{5.4} & \multicolumn{1}{c|}{5.5} & 5.8 & \multicolumn{1}{c|}{2.4} & \multicolumn{1}{c|}{3.8} & 9.7 \\ \hline\hline
\multirow{4}{*}{\textbf{ZFP(FXR)}}  & \textbf{FXR} & \multicolumn{1}{c|}{min} & \multicolumn{1}{c|}{avg} & max & \multicolumn{1}{c|}{min} & \multicolumn{1}{c|}{avg} & max & \multicolumn{1}{c|}{min} & \multicolumn{1}{c|}{avg} & max \\ \cline{2-11}
& 4 & \multicolumn{1}{c|}{8} & \multicolumn{1}{c|}{8} & 8 & \multicolumn{1}{c|}{8} & \multicolumn{1}{c|}{8} & 8 & \multicolumn{1}{c|}{8} & \multicolumn{1}{c|}{8} & 8 \\ \cline{2-11} 
 & 8 & \multicolumn{1}{c|}{4} & \multicolumn{1}{c|}{4} & 4 & \multicolumn{1}{c|}{4} & \multicolumn{1}{c|}{4} & 4 & \multicolumn{1}{c|}{4} & \multicolumn{1}{c|}{4} & 4 \\ \cline{2-11} 
 & 16 & \multicolumn{1}{c|}{2} & \multicolumn{1}{c|}{2} & 2 & \multicolumn{1}{c|}{2} & \multicolumn{1}{c|}{2} & 2 & \multicolumn{1}{c|}{2} & \multicolumn{1}{c|}{2} & 2 \\ \hline
\end{tabular}%
}
\end{table}
\begin{table} [ht]
\centering
\caption{COMPRESSION QUALITIES (PSNR)}
\label{tab:compression-quality}
\resizebox{1\columnwidth}{!}{%
\begin{tabular}{|c|c|ccc|ccc|ccc|}
\hline
\textbf{Datasets} & \textbf{} & \multicolumn{3}{c|}{\textbf{RTM}} & \multicolumn{3}{c|}{\textbf{Hurricane}} & \multicolumn{3}{c|}{\textbf{CESM-ATM}} \\ \hline\hline
\multirow{4}{*}{\textbf{SZx}}  & \textbf{ABS} & \multicolumn{1}{c|}{min} & \multicolumn{1}{c|}{avg} & max & \multicolumn{1}{c|}{min} & \multicolumn{1}{c|}{avg} & max & \multicolumn{1}{c|}{min} & \multicolumn{1}{c|}{avg} & max \\ \cline{2-11}
& 1E-2 & \multicolumn{1}{c|}{31.2} & \multicolumn{1}{c|}{37.3} & 61.5 & \multicolumn{1}{c|}{25.2} & \multicolumn{1}{c|}{26.2} & 26.8 & \multicolumn{1}{c|}{46} & \multicolumn{1}{c|}{50.6} & 55.9 \\ \cline{2-11} 
 & 1E-3 & \multicolumn{1}{c|}{40.6} & \multicolumn{1}{c|}{51.4} & 80.9 & \multicolumn{1}{c|}{40.3} & \multicolumn{1}{c|}{42} & 42.8 & \multicolumn{1}{c|}{64.2} & \multicolumn{1}{c|}{70.3} & 73.4 \\ \cline{2-11} 
 & 1E-4 & \multicolumn{1}{c|}{60.2} & \multicolumn{1}{c|}{72} & 102.6 & \multicolumn{1}{c|}{62.4} & \multicolumn{1}{c|}{63.6} & 64.3 & \multicolumn{1}{c|}{82.8} & \multicolumn{1}{c|}{93.5} & 97.2 \\ \hline\hline
\multirow{4}{*}{\textbf{ZFP(ABS)}}  & \textbf{ABS} & \multicolumn{1}{c|}{min} & \multicolumn{1}{c|}{avg} & max & \multicolumn{1}{c|}{min} & \multicolumn{1}{c|}{avg} & max & \multicolumn{1}{c|}{min} & \multicolumn{1}{c|}{avg} & max \\ \cline{2-11} 
& 1E-2 & \multicolumn{1}{c|}{36.6} & \multicolumn{1}{c|}{45.7} & 71.8 & \multicolumn{1}{c|}{33} & \multicolumn{1}{c|}{34} & 34.3 & \multicolumn{1}{c|}{52.6} & \multicolumn{1}{c|}{56.8} & 61 \\ \cline{2-11} 
 & 1E-3 & \multicolumn{1}{c|}{49.2} & \multicolumn{1}{c|}{59.5} & 88.2 & \multicolumn{1}{c|}{47.8} & \multicolumn{1}{c|}{49} & 50 & \multicolumn{1}{c|}{69.1} & \multicolumn{1}{c|}{73.2} & 77.6 \\ \cline{2-11} 
 & 1E-4 & \multicolumn{1}{c|}{69.3} & \multicolumn{1}{c|}{80.1} & 111 & \multicolumn{1}{c|}{68.5} & \multicolumn{1}{c|}{69.7} & 70.4 & \multicolumn{1}{c|}{92.2} & \multicolumn{1}{c|}{96.6} & 100.6 \\ \hline\hline
\multirow{4}{*}{\textbf{ZFP(FXR)}} \textbf{} & \textbf{FXR} & \multicolumn{1}{c|}{min} & \multicolumn{1}{c|}{avg} & max & \multicolumn{1}{c|}{min} & \multicolumn{1}{c|}{avg} & max & \multicolumn{1}{c|}{min} & \multicolumn{1}{c|}{avg} & max \\ \cline{2-11}
& 4 & \multicolumn{1}{c|}{41.9} & \multicolumn{1}{c|}{46.4} & 56.1 & \multicolumn{1}{c|}{30.9} & \multicolumn{1}{c|}{31.6} & 32.3 & \multicolumn{1}{c|}{30.3} & \multicolumn{1}{c|}{34} & 39.9 \\ \cline{2-11} 
 & 8 & \multicolumn{1}{c|}{66.7} & \multicolumn{1}{c|}{71.6} & 80.1 & \multicolumn{1}{c|}{56} & \multicolumn{1}{c|}{57.8} & 64.9 & \multicolumn{1}{c|}{58.8} & \multicolumn{1}{c|}{60.5} & 62.9 \\ \cline{2-11} 
 & 16 & \multicolumn{1}{c|}{113.8} & \multicolumn{1}{c|}{118.5} & 127.3 & \multicolumn{1}{c|}{103.7} & \multicolumn{1}{c|}{105.5} & 113 & \multicolumn{1}{c|}{106.4} & \multicolumn{1}{c|}{108.1} & 110.6 \\ \hline
\end{tabular}%
}
\end{table}

In addition to the above analysis, we also compare SZx, ZFP(ABS), and ZFP(FXR) regarding the compression throughput, ratio, and quality, using different datasets with various error bounds or rates, as shown in Table \ref{tab:compression-speed}, \ref{tab:COMPRESSION-RATIOS}, and \ref{tab:compression-quality}. To ensure a fair evaluation, all compression and decompression processes were executed using a single thread on an Intel Xeon E5-2695v4 CPU. We adopt the 1D compression mode in that the dimensional information will have to be skipped due to the 1D chunk-wise design in most of the MPI collectives. The information about the datasets we used here is detailed in Table \ref{tab:datasets}, which will be found in Section \ref{sec:setup}. Because of the space limit, we use the QVAPORf and CLOUD fields in the Hurricane and CESM-ATM dataset, respectively, in our experiments, and other fields exhibit very similar results. We observe that SZx is much faster than ZFP(ABS) by up to 4.1$\times$ in compression and 5.7$\times$ in decompression with similar compression ratios and qualities. Additionally, with similar compression qualities, ZFP(FXR) has the lowest compression speed and compression ratio compared to both ZFP(ABS) and SZx. Therefore, we develop our customized compressor based on SZx in terms of the context of MPI collectives. For the purpose of comparison, we also implement compression-enabled point-to-point communication-based collectives based on both ZFP(FXR) and ZFP(ABS), which serve as baselines.

\subsection{Characterization of Performance Bottlenecks}
\label{sec:Characterization}

We integrate SZx into point-to-point communication of the ring-based allreduce algorithm, in order to understand the key bottlenecks of the collective performance, which will be a fundamental work to guide our optimization strategies. In this characterization, we use ring-based allreduce which serves as a very good example, because it consists of both collective data movement (allgather) and collective computation (reduce\_scatter). Figure \ref{fig-original-naive} presents the detailed performance breakdown for direct integration of SZx under a series of data sizes. Our analysis reveals that in the original ring-based allreduce algorithm, the all-gather operation accounts for approximately 60\% of the overall execution time, as the communications are not overlapped at all with each other, unlike the reduce-scatter stage. The Wait operation is the second most time-consuming one, which waits for the completion of non-blocking send and receives before conducting the reduction operation. The remaining operations in the original allreduce algorithm include Memcpy (local data copies), Reduction (Reduce operations), and Others (data allocations and other calculations), which together constitute about 20\% of the overall execution time. After integrating the SZx, the bottleneck turns into the ComDecom (compression and decompression) as the data needs to be compressed before being sent out and needs to be decompressed every time the receiver receives them. We also observe a reduction in Allgather and Wait times, suggesting a decrease in transferred data, and both of the MPI-related time can be further optimized. However, the Others part also takes a significant amount, specifically 23\% in the 278MB case. This is because the SZx requires users to free compression-generated buffers after the compressor is called, resulting in a significant overhead.


\begin{figure}[ht]
    \centering
    {\includegraphics[width=0.6\linewidth]{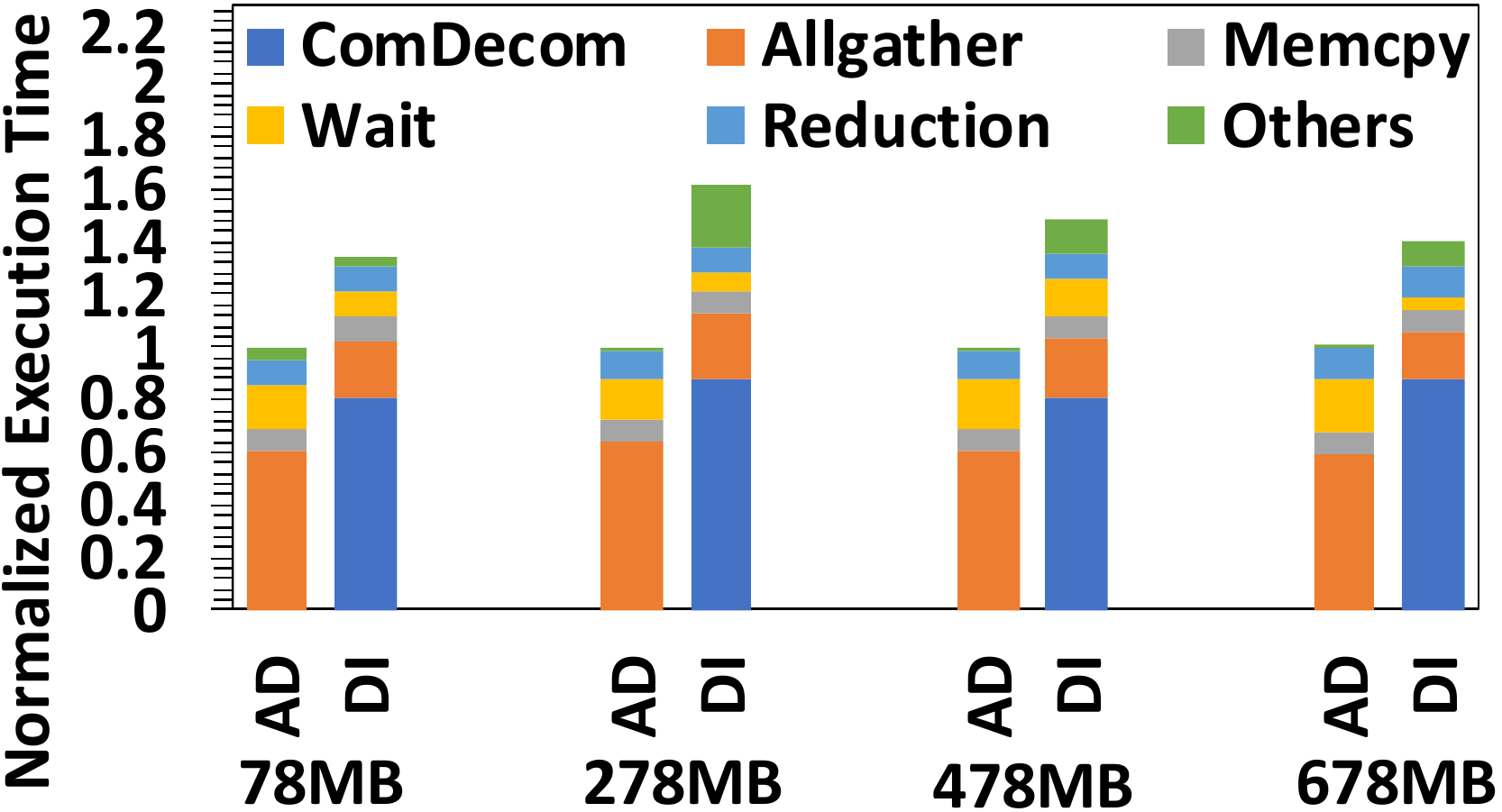}}
    \caption{Compare the performance of AD(Original MPI\_Allreduce) and the DI(Direct Integration) of SZx from 78MB to 678MB with a 200MB step.} 
    \label{fig-original-naive}
\end{figure}

\subsection{Step-wise Optimizations}

In this subsection, we detail our step-wise optimization strategies, a principal contribution of this paper. To facilitate explanation, we present our implementation of a ring-based allreduce algorithm that incorporates lossy compression. These optimization strategies are applicable to other collective operations as well. Notably, in the ring-based allreduce, the data transfer required for each process is just $\frac{2(N-1)}{N}$$\cdot$$D_{input}$, where $D_{input}$ represents the input data size and $N$ is the number of processes. Thus, this design is highly efficient for processing long messages. Furthermore, our integration of lossy compression significantly enhances the efficiency of collective communication involving long messages. We have termed our lossy-compression-enabled MPI\_Allreduce as C-Allreduce, with `C' denoting compression. Subsequently, we will discuss the design and implementation specifics of C-Allreduce in a structured, step-by-step approach.

\subsubsection{\textbf{Utilize our collective data movement framework}}
\label{sec:Utilize-data-move}
To reduce the compression overhead and balance communications, we utilize the data movement framework that we presented in Section \ref{sec-data-movement-framework}. At the beginning, every process compresses its local data and stores the compressed data size. Then, every process synchronizes with each other to collect the compressed data sizes in a local integer array $compressed\_sizes$. As the compressed data size only has four bytes, this step is very fast. After that, all processes get the sum of all the compressed data sizes, noted as $total\_count$. Then, each process communicates with each other with a fixed pipeline size until every process has sent $to\_send = total\_count - compressed\_sizes[send\_rank]$ and received $to\_recv = total\_count - compressed\_sizes[self\_rank]$ from other processes in a ring communication pattern. After all communications end, every process starts to decompress all the received compressed data and store the decompressed data in the receive buffer. Note that they do not need to decompress the data that are compressed by themselves. After this step, we can significantly decrease the time spent by compression and allgather communications compared with the direct integration of SZx. Besides, our solution can also preserve the quality/accuracy of the data very well because of the error-bounding feature, which will be demonstrated later in Section \ref{sec-image-stacking}.

\subsubsection{\textbf{Customize SZx to reduce communication overhead with our collective computation framework}}
\label{sec:Customize-SZx-reduce}
In order to use our collective computation framework in the reduce-scatter stage of the ring-based allreduce algorithm, we need to redesign the compression workflow of SZx so that we can consistently poll the progress of the Isend and Irecv inside of the compression and decompression. Therefore, we design and implement the PIPE-SZx (pipelined SZx) based on the original SZx. Instead of compressing the original data as a whole, we divide the compression process into small chunks, each of which handles 5120 data points. Between the compression of two adjacent chunks, we actively poll the communication progress of the non-blocking receive. However, the compressed data of each chunk cannot be simply combined together, otherwise the compressed data cannot be correctly decompressed because each compressed chunk is of variable uncertain length. To solve this problem, we decide to store the compressed data of all chunks in the same output buffer and pre-allocate enough memory space (four bytes per chunk, small memory consumption) at the front of the buffer for storing the compressed data sizes of those chunks together (essentially a kind of index), instead of storing them along with the compressed data chunks. Such a design is more cache-friendly, thus having lower overhead. During the decompression, we maintain a chunk-starting-location pointer based on the recorded compressed chunk sizes to tell the algorithm where the decompression operation should start for each chunk. 
We repeat this process chunk by chunk and poll the progress of the non-blocking send between decompression chunks. Through this optimization, we can hide the communication in the reduce-scatter stage inside of compression, which further improves the performance of our C-Allreduce design.

\section{Experimental Evaluation}\label{exp-setup-sec}

In this section, we present and discuss the evaluation results.

\subsection{Experimental Setup}
\label{sec:setup}

Since inter-node communication is the major bottleneck for collectives as discussed previously, we utilized a 128-node cluster with one process per node in our experiments. Each node is equipped with two Intel Xeon E5-2695v4 Broadwell processors. Furthermore, each NUMA node contains 64 GB of DDR4 memory, resulting in a total of 128 GB of memory per node. The nodes are interconnected via Intel Omni-Path Architecture (OPA), providing a maximum message rate of 97 million per second and a bandwidth of 100 Gbps.

\begin{table}[ht]
\centering
\caption{Information of The Scientific Datasets}
\resizebox{1\columnwidth}{!}{%
\begin{tabular}{|c|c|c|c|}
\hline
{\textbf{Applications}} & {\textbf{\# files}} & {\textbf{Dimensions}} & {\textbf{Descriptions}} \\ \hline\hline
{RTM\cite{Kayum2020RTM}} & \multirow{1}{*}{70} & \multirow{1}{*}{849$\times$849$\times$235} & \multirow{1}{*}{Seismic Wave} \\ \hline
\multirow{1}{*}{Hurricane\cite{hurricane2004}} & \multirow{1}{*}{48$\times$13} & \multirow{1}{*}{100$\times$500$\times$500} & \multirow{1}{*}{Weather Simulation} \\\hline
\multirow{1}{*}{CESM-ATM\cite{Zhao2020SDRBench}} & \multirow{1}{*}{26$\times$33} & \multirow{1}{*}{1800$\times$3600} & \multirow{1}{*}{Climate Simulation}\\ \hline
\end{tabular}%
}
\label{tab:datasets}
\end{table}

MPI collectives are common operations used in the simulation analysis. For instance, generating stacking images in RTM (essentially an allreduce sum operation) is a typical real-world example \cite{Gurhem2021Kirchhoff}, which will be demonstrated at the end of this section. We also utilize various datasets from a series of applications, including the reverse time migration (RTM) dataset, Hurricane dataset, and CESM-ATM dataset, to evaluate our solution. 
Table \ref{tab:datasets} shows the detailed specifications of these datasets. Our baselines consist of original MPI collectives in MPICH 4.1.1, MPI collectives implemented by compression-enabled point-to-point communications with SZx, fix-accuracy ZFP, and fix-rate ZFP (version 0.5.5). For our experiments, we adopt a two-stage approach, including a warm-up stage and an execution stage. We conduct 10 runs for each stage and report the average results to present the general performance.

\subsection{Step-wise Optimizations to C-Allreduce with Performance Analysis}
\label{sec:evaluation}
In this section, we carry out optimizations to our C-Allreduce (\textit{C-Coll} enhanced Allreduce) integrated with SZx step by step and show the performance on 16 Broadwell nodes. The ring-based allreduce that we implement contains a reduce-scatter stage and an allgather stage. Thus, we breakdown the total execution time into several parts: ComDecom (i.e., the time needed to compress and decompress data), Allgather (i.e., the time required to transfer data in the allgather stage), Memcpy (i.e., the time spent on coping data in the reduce-scatter stage), Wait (i.e., the non-overlapped time spent on transferring data in the reduce-scatter stage), Reduction (i.e., the time required to do reduction operation), and Others (i.e., the time needed to do other data allocation and calculations). We adopt the RTM dataset to demonstrate our step-wise optimizations, which is the largest one among all three representative scientific datasets we evaluate. All the variants benchmarked in the following contexts are summarized in Table \ref{tab:different-optimizations}.

\begin{table}[]
\centering
\caption{Step-wise Optimizations of C-Allreduce}
\label{tab:different-optimizations}
\resizebox{\columnwidth}{!}{%
\begin{tabular}{|l|l|}
\hline
\textbf{Method (Abbr.)} & \textbf{Description about the Allreduce implementation} \\ \hline \hline
Original MPI\_Allreduce (AD) &  No compression \\ \hline
Direct Integration (DI) & Implemented with CPRP2P \\ \hline
Novel Design (ND) & Optimized by our collective data movement framework \\ \hline
Overlapped Optimization (Overlap) & Optimized by our collective computation framework \\ \hline
\end{tabular}%
}
\end{table}

\subsubsection{\textbf{Evaluating our collective data movement framework}}
\label{sec:eval-framework}

Figure \ref{fig-new-design} shows the performance improvement achieved by our novel design in the allgather stage for data sizes ranging from 28MB to 678MB. This Novel Design (ND) leverages our collective data movement framework, resulting in a considerable reduction in both compression and decompression time. Notably, at the data size of 128MB, the ND achieves a speed-up of up to 1.48$\times$ when compared to the Direct Integration (DI) approach discussed in Section \ref{sec:Characterization}. Additionally, our balanced communication in Allgather using the ND is up to 7.1$\times$ faster than the unbalanced communication in DI at 628MB. We analyze the key reason why our solution can obtain a significant performance improvement as follows. In fact, to overcome the compression bottleneck and balance MPI communications, we utilize our collective data movement framework that pre-compresses the data before transmission and decompresses it after all communications, rather than using expensive compression-enabled point-to-point communication (CPR-P2P) in collective routines. This novel design can significantly reduce the amount of compression required during collective communication. Using CPR-P2P also brings unbalanced communications as the compressed data sizes may vary, but we can balance the communications with a fixed pipeline size in our new design because we do not need to compress the data every time before we send it. 

Note that using CPR-P2P may accumulate errors during intensive collective communication, such as ring-based communication, as the same data is repeatedly passed from one process to another. Therefore, we have utilized our new framework to ensure that errors in the final results are bounded, which we discuss in the application evaluation section \ref{sec-image-stacking}.
\begin{figure}[ht]
    \centering
    {\includegraphics[width=0.95\linewidth]{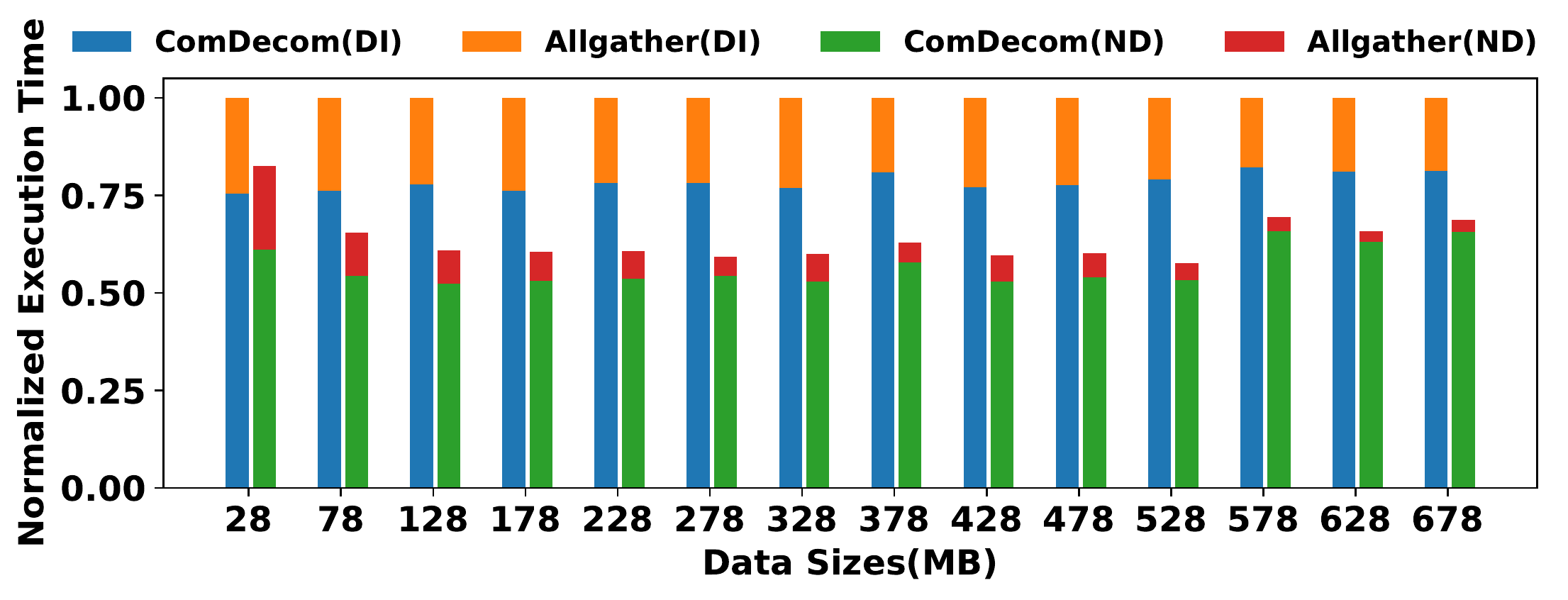}}
    \caption{Compare the performance of DI (Direct Integration) and our ND (Novel Design) from 28MB to 678MB.} 
    \label{fig-new-design}
\end{figure}

\subsubsection{\textbf{Evaluating reduced communication overhead with our collective computation framework}}
We demonstrate the effectiveness of our collective computation framework in this section. From Figure \ref{fig-Overlap}, it is evident that the carefully overlapped version (Overlap) leads to significantly less Wait time compared with the previous version (ND), resulting in a performance boost of up to 4.9$\times$ for the data size of 678MB. The rationale for this performance improvement is shown in the following text. To utilize our collective computation framework for hiding communication during compression, we design and implement PIPE-SZx (pipelined SZx), which could break the compression process into small chunks and allow us to overlap the compression with communication in a fine-grained pipelined manner. As a result, we can significantly reduce MPI\_Wait time in the reduce-scatter stage, which uses non-blocking communications.

\begin{figure}[ht]
    \centering
    {\includegraphics[width=0.8\linewidth]{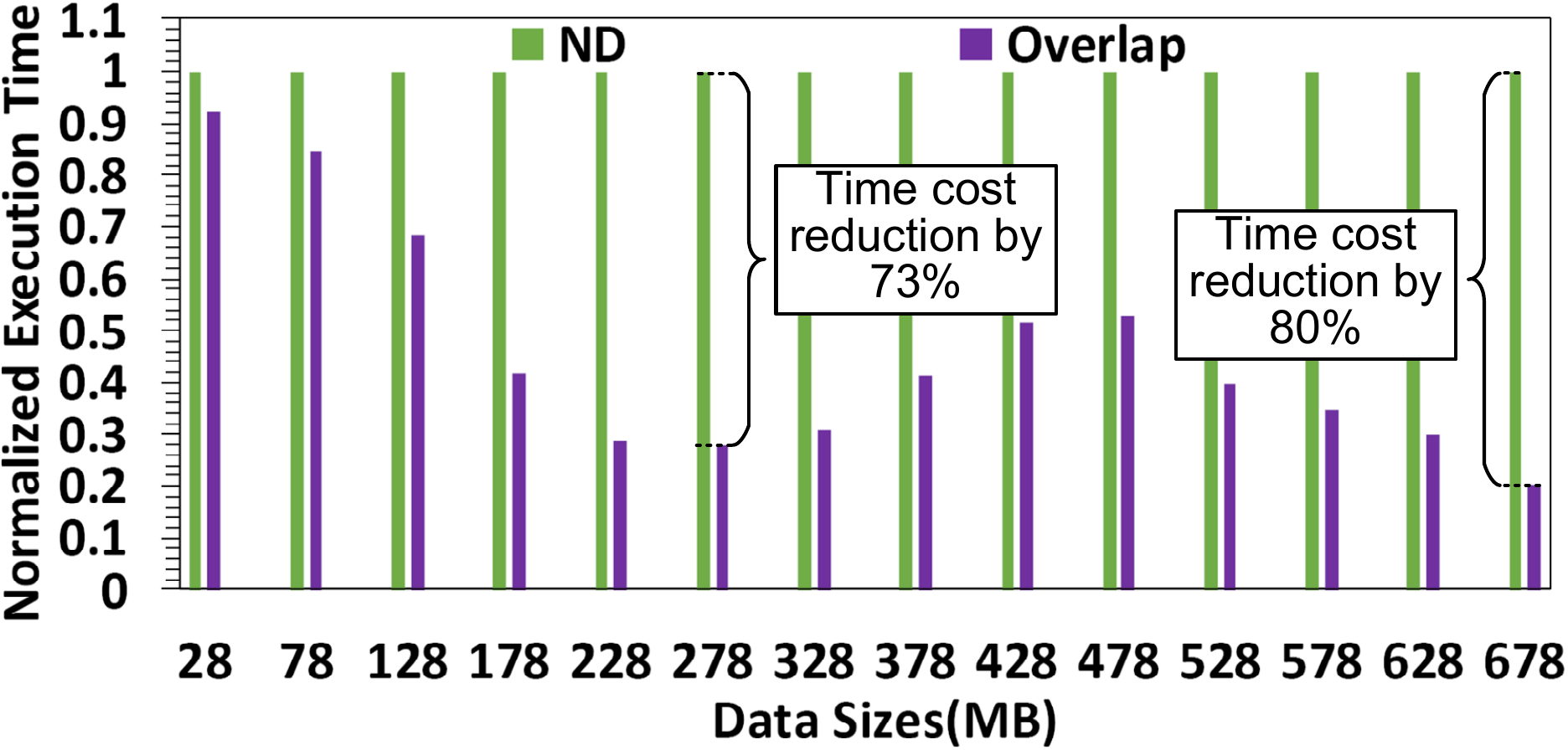}}
    \caption{Compare the wait time of ND (Novel Design) and Overlap (Overlap Optimization) from 28MB to 678MB with a 50MB step.} 
    \label{fig-Overlap}
\end{figure}

\subsubsection{\textbf{Evaluating overall running time of different optimizations}}
In Figure \ref{fig-step-wise-overall}, we compare the end-to-end execution time of our compression-enabled Allreduce with step-wise optimizations against the Original MPI\_Allreduce (AD). We can see that the end-to-end performance of our compression-enabled Allreduce gradually increase with these optimizations and it consistently beats the AD across various message sizes with our Overlapped Optimization (Overlap). We call the fully optimized variant as C-Allreduce and will thoroughly evaluate both the performance and accuracy of it in the following sections.

\begin{figure}[ht]
    \centering
    {\includegraphics[width=0.8\linewidth]{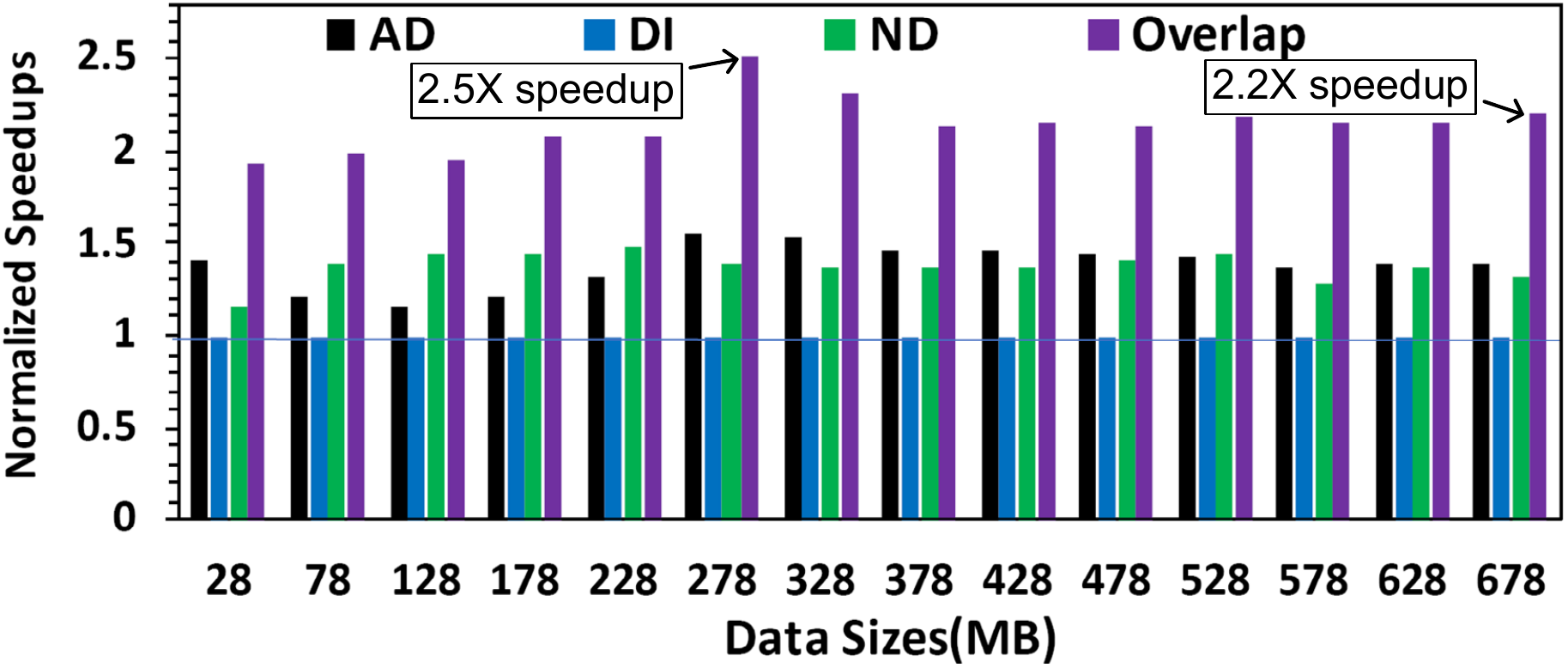}}
    \caption{Compare the overall running time of different optimizations from 28MB to 678MB with a 50MB step.}
    \label{fig-step-wise-overall}
\end{figure}

\subsection{End-to-end Comparisons of C-Allreduce with Baselines}
In this section, we compare the performance of our C-Allreduce with four different baselines on various data sizes, node numbers, and datasets. The baselines include the original MPI\_Allreduce without compression (Allreduce), MPI\_Allreduce implemented by compression-enabled point-to-point communication (CPR-P2P) with fixed-rate ZFP (ZFP(FXR)), fixed-accuracy ZFP (ZFP(ABS)), and SZx. The SZx and ZFP(ABS) have a fixed accuracy of 1E-3, and the ZFP(FXR) has a rate of 4, if not specifically mentioned.

\begin{figure}[ht]
    \centering
    {\includegraphics[width=0.8\linewidth]{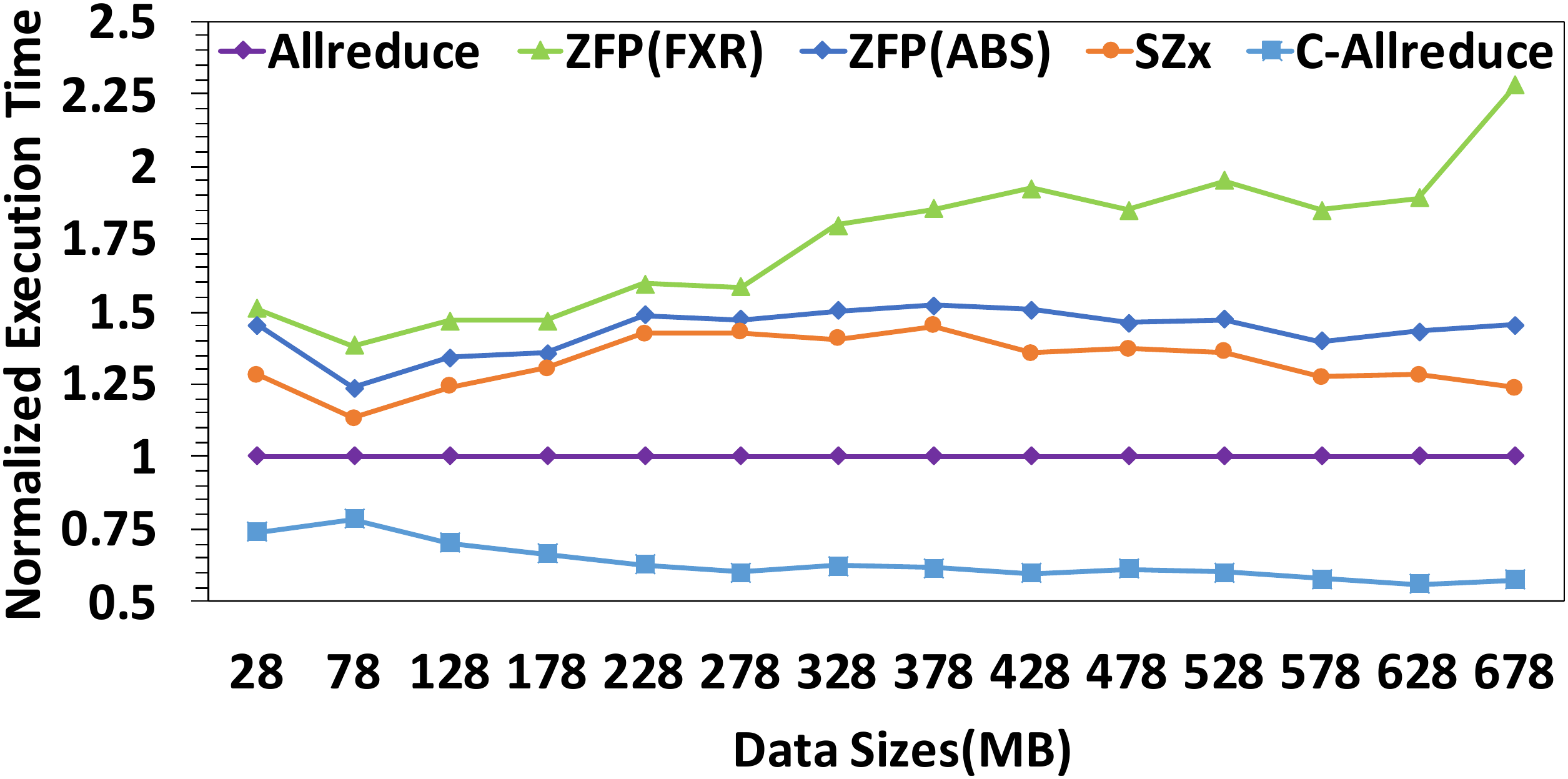}}
    \caption{Compare the performance of our C-Allreduce and multiple baselines from 28MB to 678MB with a 50MB step on 128 nodes.} 
    \label{fig-128-sizes}
\end{figure}

\subsubsection{\textbf{Evaluating with different data sizes on 128 nodes}}\label{sec-128-sizes}
In this section, we present the performance of our C-Allreduce, along with related baselines, using data sizes ranging from 28MB to 678MB on 128 Broadwell nodes. The RTM dataset is adopted in these experiments. From Figure \ref{fig-128-sizes}, we can observe that the SZx-integrated baseline performs the best among the three compression-integrated baselines, and ZFP(FXR) has the worst performance. This is because SZx has a faster compression speed compared to ZFP(FXR) and ZFP(ABS). Additionally, ZFP(ABS) and SZx have better compression ratios than ZFP(FXR) as we previously mentioned in \ref{sec:high-speed-compressor}. However, none of the three CPR-P2P baselines can outperform the original Allreduce. By employing our novel framework and step-wise optimizations, our C-Allreduce achieves significantly higher performance than the original Allreduce with up to 1.8$\times$ performance improvement.

\begin{figure}[ht]
    \centering
    {\includegraphics[width=0.8\linewidth]{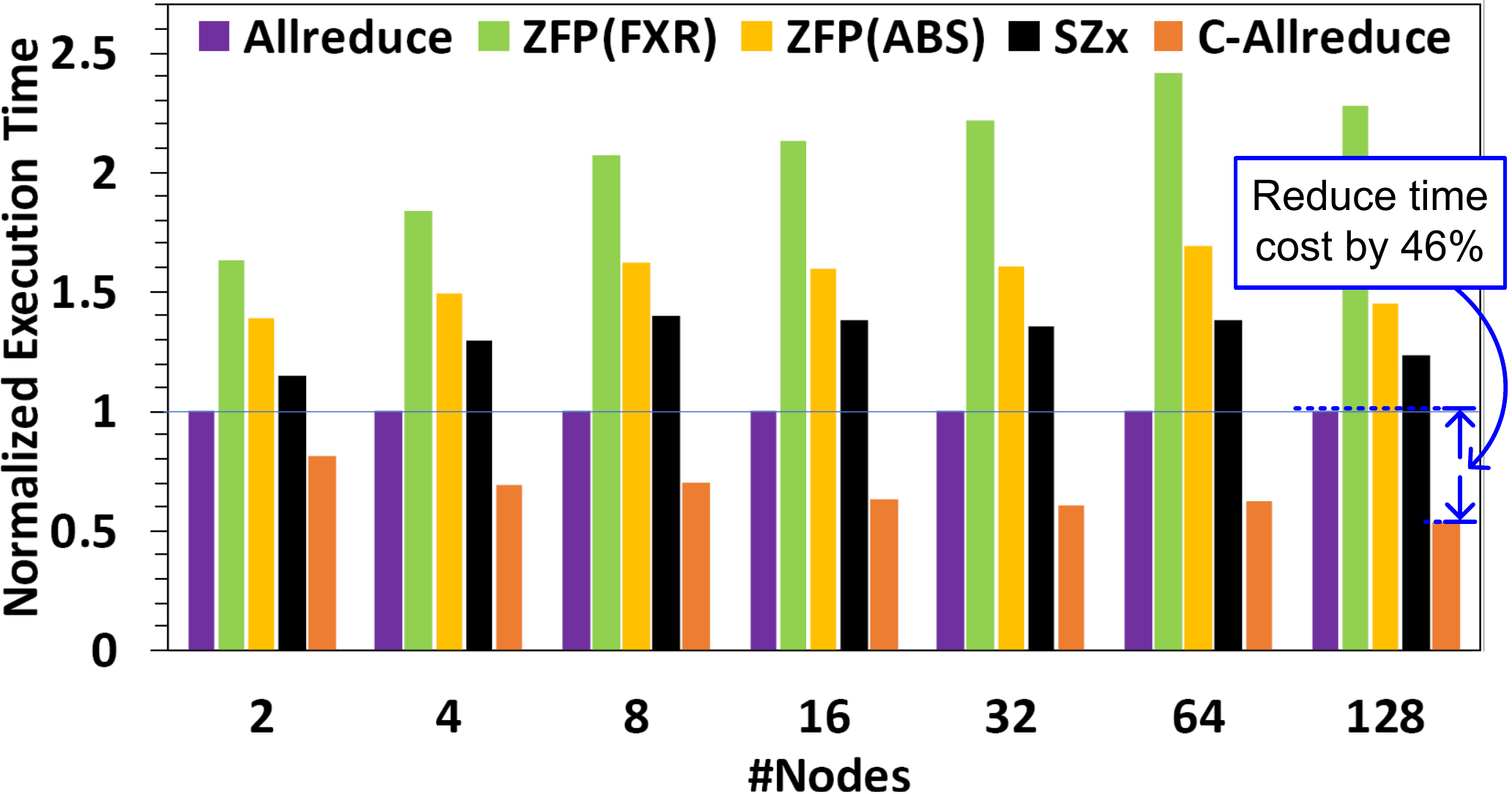}}
    \caption{Compare the performance of our C-Allreduce and multiple baselines from 2 to 128 nodes with 678MB message.} 
    \label{fig-128-nodes}
\end{figure}

\subsubsection{\textbf{Evaluating with different node numbers with 678MB data}}
To demonstrate the scalability of our approach, we compare the normalized execution time of our C-Allreduce and four different baselines using a fixed data size of 678MB across 2 to 128 nodes. The RTM dataset is adopted in these experiments. As shown in Figure \ref{fig-128-nodes}, our C-Allreduce outperforms all the baselines across various node numbers. It can reach a performance boost of up to 1.8$\times$ compared to the original Allreduce. Similar to our observation in Section \ref{sec-128-sizes}, we found that all compressor-integrated baselines exhibit performance degradation compared to Allreduce, with the SZx-integrated baseline performing the fastest among them.

\subsubsection{\textbf{Evaluating with different application datasets}}

\begin{table}[ht]
\centering
\caption{Compression Ratios}
\label{tab:SZx-diiff-data}
\resizebox{0.8\columnwidth}{!}{%
\begin{tabular}{|c|cccc|c|}
\hline
\textbf{Datasets} & \multicolumn{3}{c|}{\textbf{Hurricane}} & \multicolumn{1}{l|}{\textbf{CESM-ATM}} \\ \hline\hline
\textbf{Fields} & \multicolumn{1}{c|}{PRECIPf} & \multicolumn{1}{c|}{QGRAUPf} & \multicolumn{1}{c|}{CLOUDf} & Q \\ \hline
\textbf{CPR} & \multicolumn{1}{c|}{33.8} & \multicolumn{1}{c|}{58.3} & \multicolumn{1}{c|}{39.9} & 79.1 \\ \hline
\end{tabular}%
}
\end{table}

In this section, we evaluate the performance of our C-Allreduce and related baselines on different datasets using the error bound of 1E-4. Due to the page limitation, we only show the SZx baseline here as it has the best performance among compression-enabled baselines. The PRECIPf, QGRAUPf, and CLOUDf fields are from the Hurricane application dataset and the Q field is originated from the CESM-ATM application dataset. They have different data distributions and averaged compression ratios as shown in Table~\ref{tab:SZx-diiff-data}. From Figure~\ref{fig-diff-datasets}, we can notice that our C-Allreduce has the best performance among all the implementations. Specifically, 1.74$\times$ and 1.58$\times$ speedups can be achieved in QGRAUPf and PRECIPf fields, respectively. Such consistent high performance of C-Allreduce across varied data fields highlights its efficiency and adaptability. On the contrary, the fastest compression-enabled baseline (SZx) still has performance degradation compared with the original Allreduce due to the limited performance of the CPR-P2P in all cases.
\begin{figure}[ht]
    \centering
    {\includegraphics[width=0.7\linewidth]{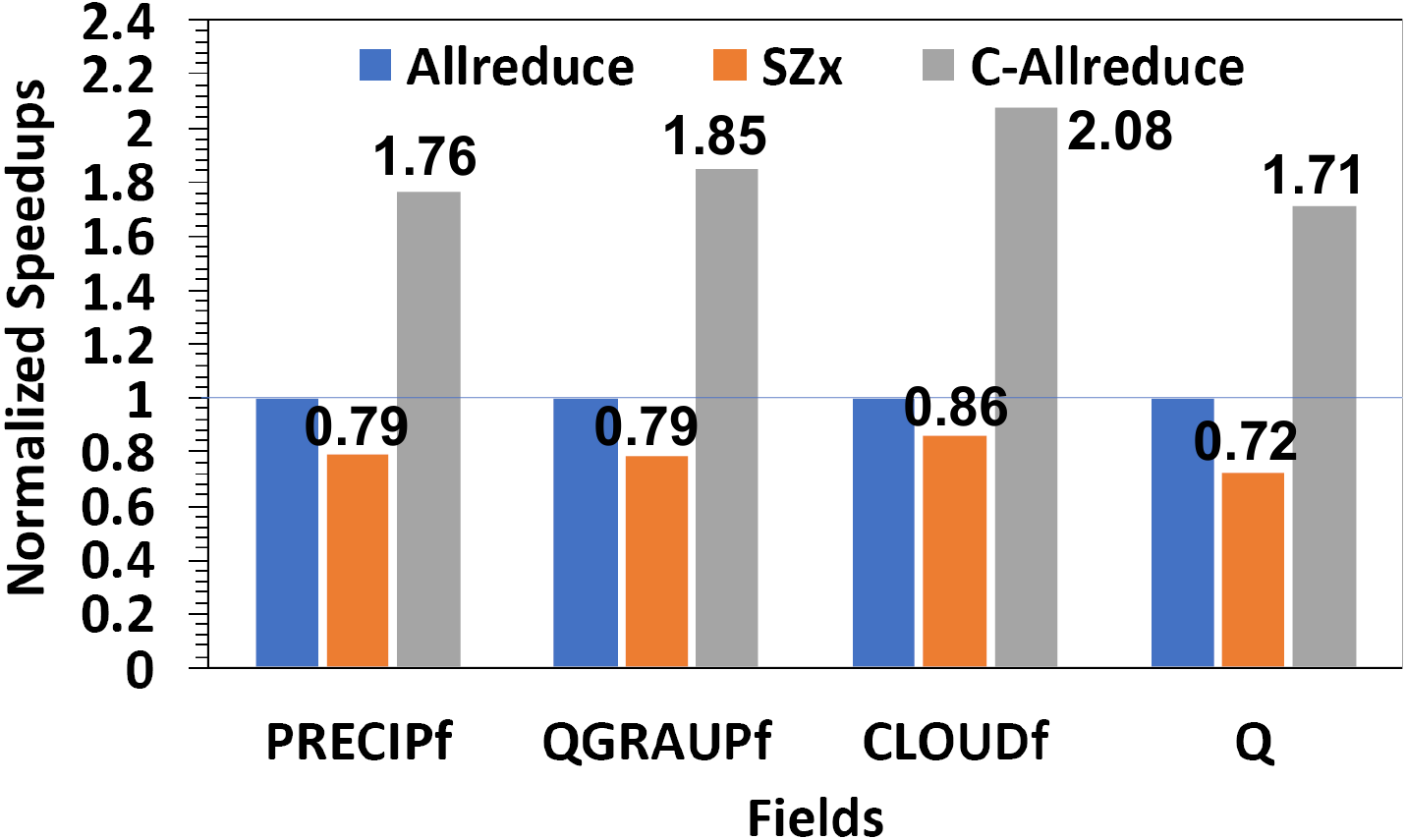}}
    \caption{Compare the performance of our C-Allreduce and multiple baselines in different application datasets.} 
    \label{fig-diff-datasets}
\end{figure}

Apart from the performance, we also evaluate the accuracy using visualization and numerical metrics including the widely-used peak signal-to-noise ratio (PSNR) \cite{PSNR} and normalized root mean squared error (NRMSE) \cite{shcherbakov2013errormetrics}. Figure \ref{fig-hurricane-quality} represents the visual and numerical evaluation of our C-Allreduce with the Hurricane application dataset. The excellent image quality, along with the great PSNR(60.04) and NRMSE(1E-3) demonstrate that our C-Allreduce delivers a well-controlled accuracy. The same phenomenon can be witnessed in the Figure \ref{fig-CESM-quality}, where we evaluate the accuracy of C-Allreduce with the CESM-ATM application dataset. Thus, we can conclude that the selected error bound is suitable for other applications like Hurricane and CESM-ATM.

\begin{figure}[ht]
    \centering
    \subfloat[Allreduce w/o compression]{\includegraphics[width=.4\linewidth]{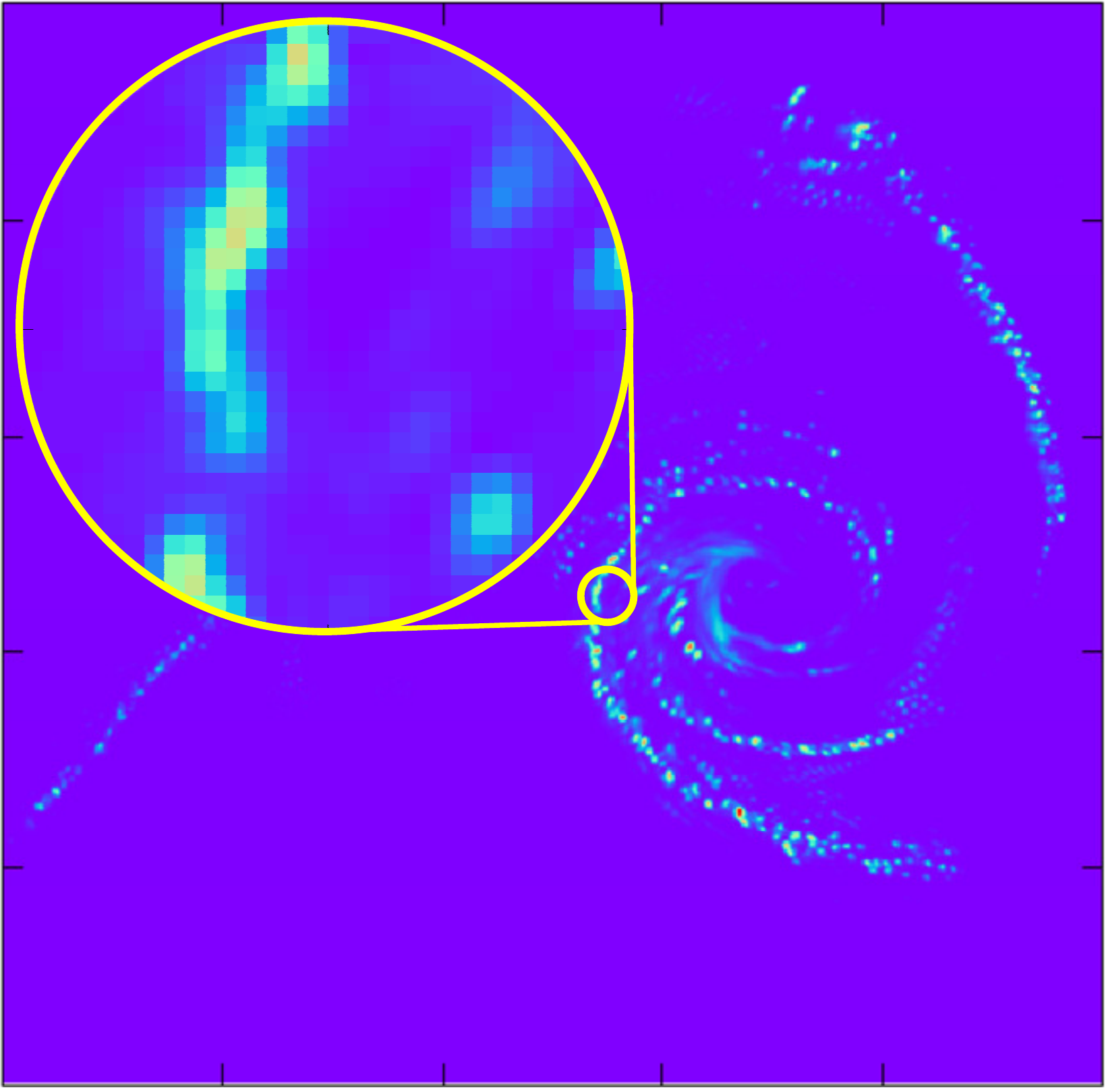}}\hspace{3mm}
    \subfloat[C-Allreduce]{\includegraphics[width=.403\linewidth]{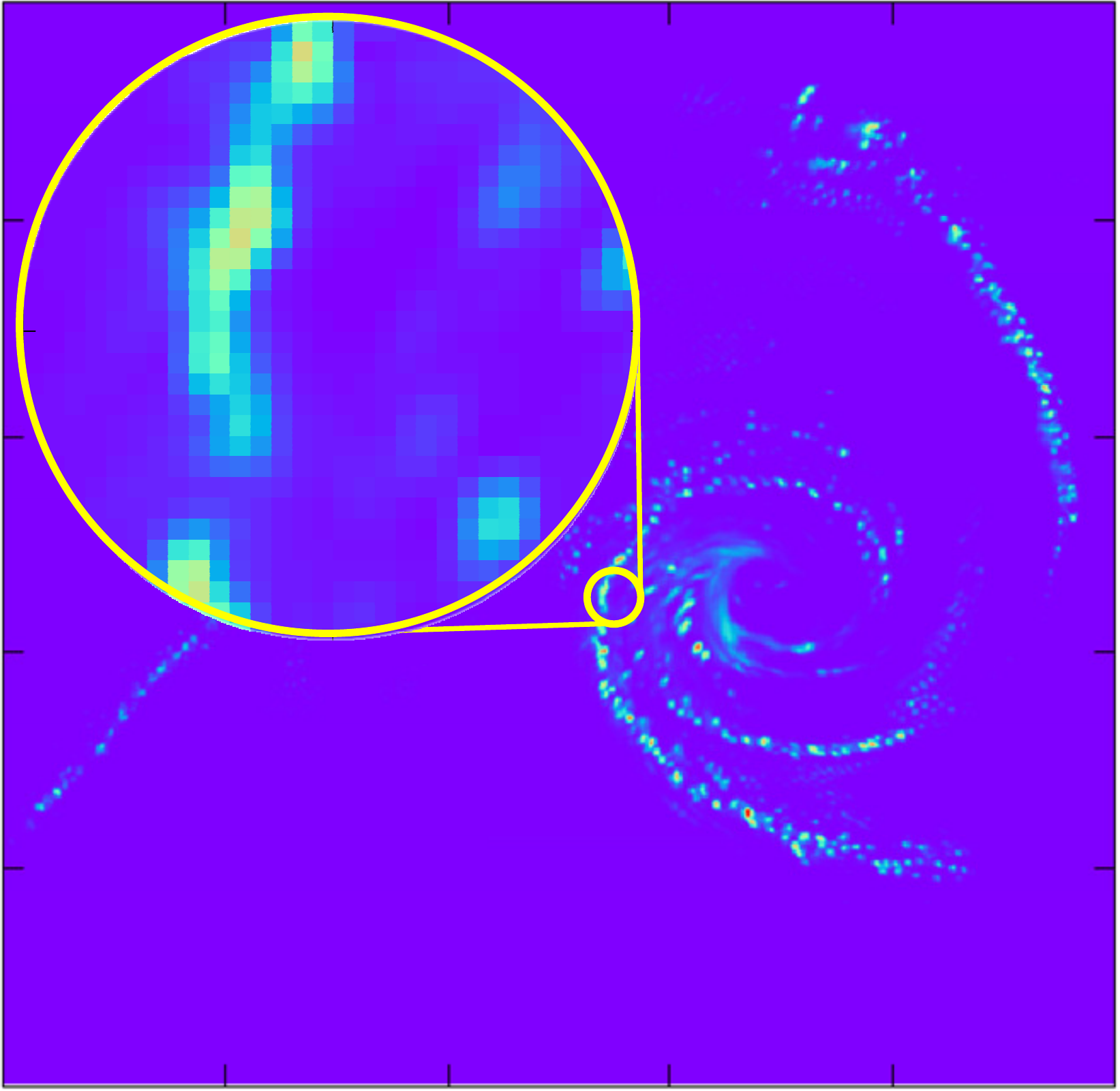}}

    \caption{The visualization and numerical evaluation of the accuracy of our C-Allreduce using Hurricane application dataset. Here the PSNR and NRMSE of the sub-figure(b) are 60.04 and 1E-3, respectively.}
    \label{fig-hurricane-quality}
    \vspace{-2mm}
\end{figure}

\begin{figure}[ht]
    \centering
    \subfloat[Allreduce w/o compression]{\includegraphics[width=.4\linewidth]{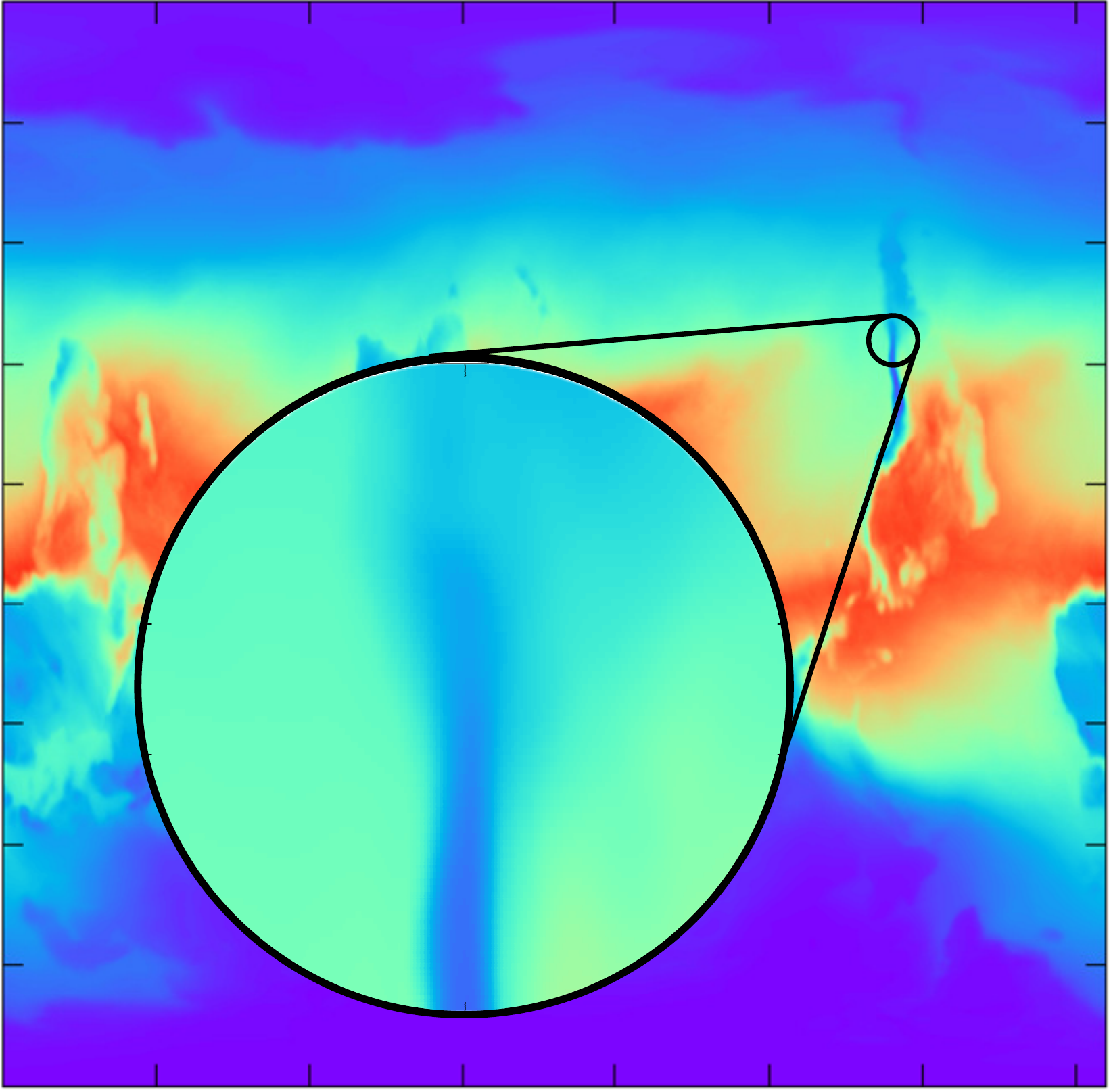}}\hspace{3mm}
    \subfloat[C-Allreduce]{\includegraphics[width=.4\linewidth]{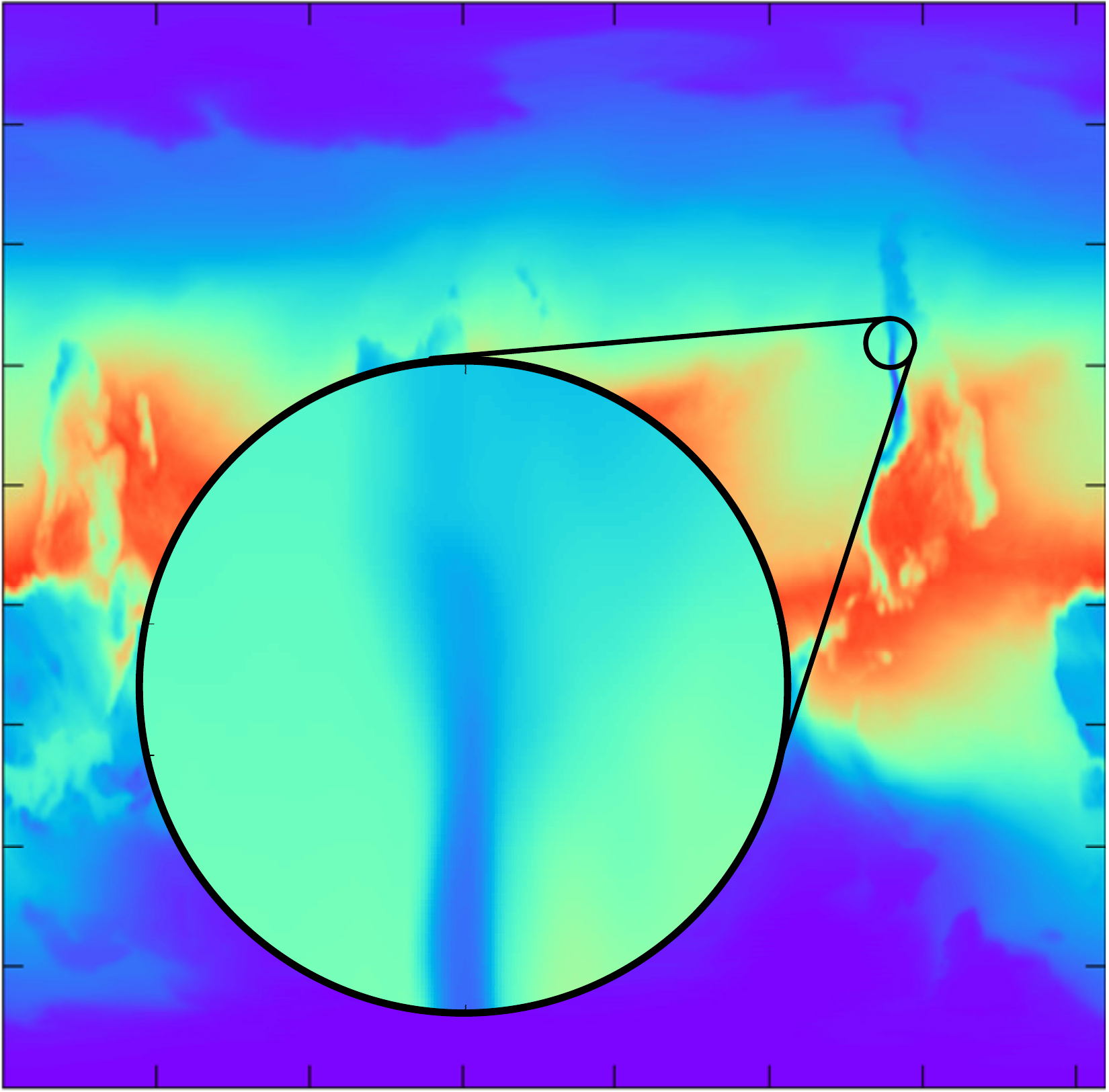}}

    \caption{The visualization and numerical evaluation of the accuracy of our C-Allreduce using CESM-ATM application dataset. Here the PSNR and NRMSE of the sub-figure(b) are 59.19 and 1E-3, respectively.}
    \label{fig-CESM-quality}
    \vspace{-2mm}
\end{figure}

\subsection{Generalizability Demonstration on Other MPI Collectives}
We have demonstrated the high performance of our C-Allred\linebreak uce, consisting of C-Allgather and C-Reduce-scatter. To showcase the generalizability of our frameworks and optimizations, we also present C-Bcast and C-Scatter, which utilize the ubiquitous binomial tree algorithm adopted by MPICH. We conduct experiments on the RTM dataset ranging from 28MB to 678MB using 16 Broadwell nodes. In Figure \ref{fig-portability}, we present the speedups of our C-Scatter and C-Bcast, normalized against the original MPI\_Scatter and MPI\_Bcast (Baseline), respectively. We also compare our C-Scatter and C-Bcast with the SZx-integrated CPR-P2P baselines (SZx). Our C-Scatter is 1.8$\times$ faster than Baseline, while our C-Bcast has up to 2.7$\times$ speedup compared to Baseline. These performance improvements originated from the decreased data transfer volume and diminished compression overheads from our proposed frameworks. The speedups are even more significant compared to our C-Allreduce because collective data movement benefits more from our frameworks than collective computation. However, the SZx-integrated CPR-P2P baselines are much slower than the baseline due to significant compression overheads.

\begin{figure}[ht]
    \centering
    {\includegraphics[width=0.8\linewidth]{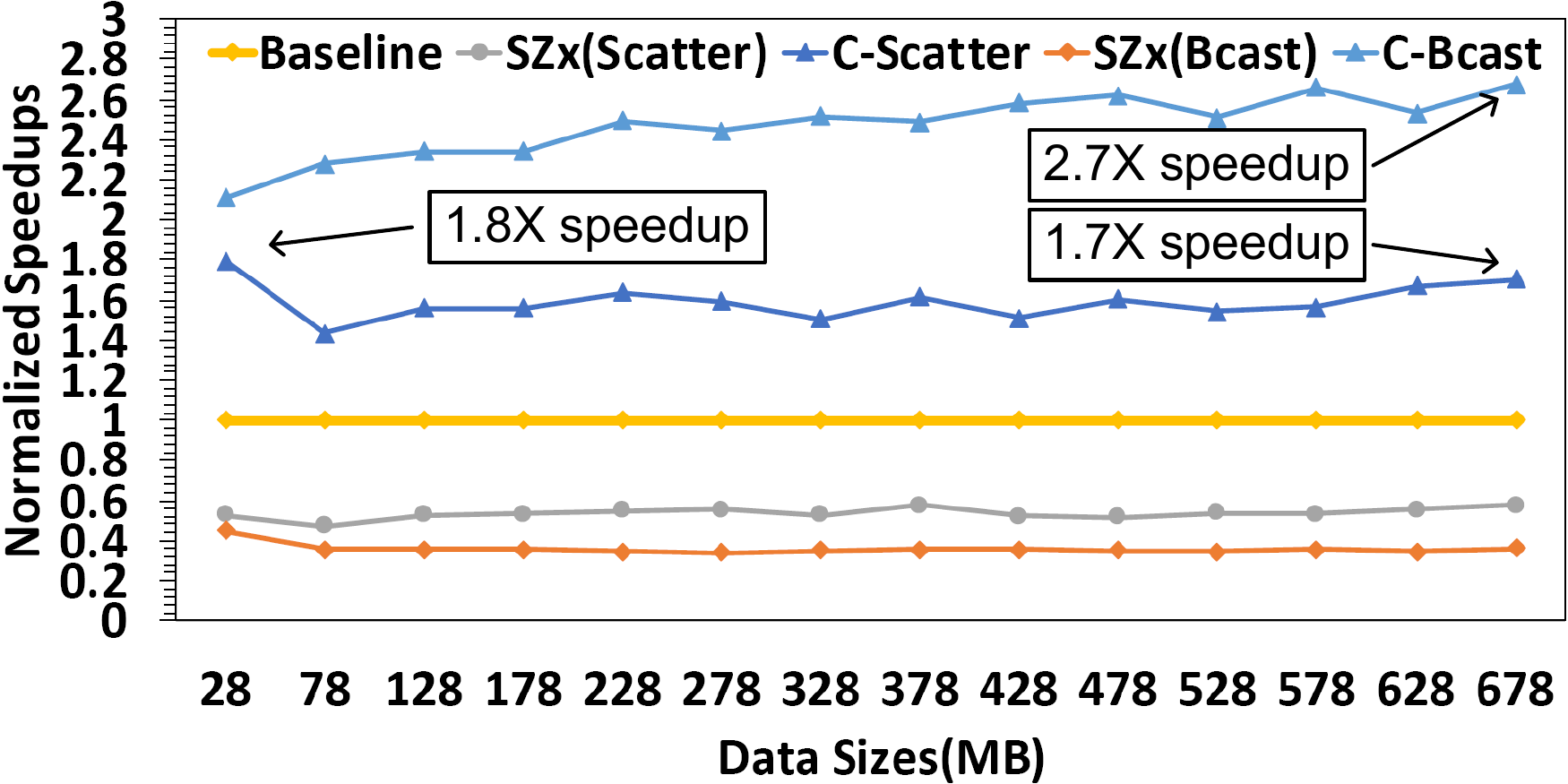}}
    \caption{Generalizability demonstration of our proposed framework and optimizations with MPI\_Scatter and MPI\_Bcast from 28MB to 678MB with a 50MB step.} 
    \label{fig-portability}
\end{figure}

\subsection{Evaluation of Image Stacking Performance and Accuracy}
\label{sec-image-stacking}
Image stacking is a well-known technique used in various scientific domains, including climate simulation and geology, to generate high-quality images by combining different images. For example, MPI is used by researchers to sum the individual images into final images \cite{Gurhem2021Kirchhoff}. In this experiment, we conduct image stacking of the RTM dataset on 16 nodes. As each snapshot has different value ranges, we use the fixed-accuracy (ABS) mode to compress the data so that each snapshot contributes a similar amount of errors rather than letting the snapshots with large value ranges dominate the errors. Three absolute error bounds (i.e., 1E-2, 1E-3, and 1E-4) are selected to demonstrate the flexibility between the accuracy and performance of our C-Allreduce. The same error bounds are used for the baseline SZx and ZFP(ABS). For ZFP(FXR), three fixed rates (4, 8, and 16) are selected and the related compression ratios are 8, 4, and 2. 

\begin{figure}[ht]
    \centering
    {\includegraphics[width=0.7\linewidth]{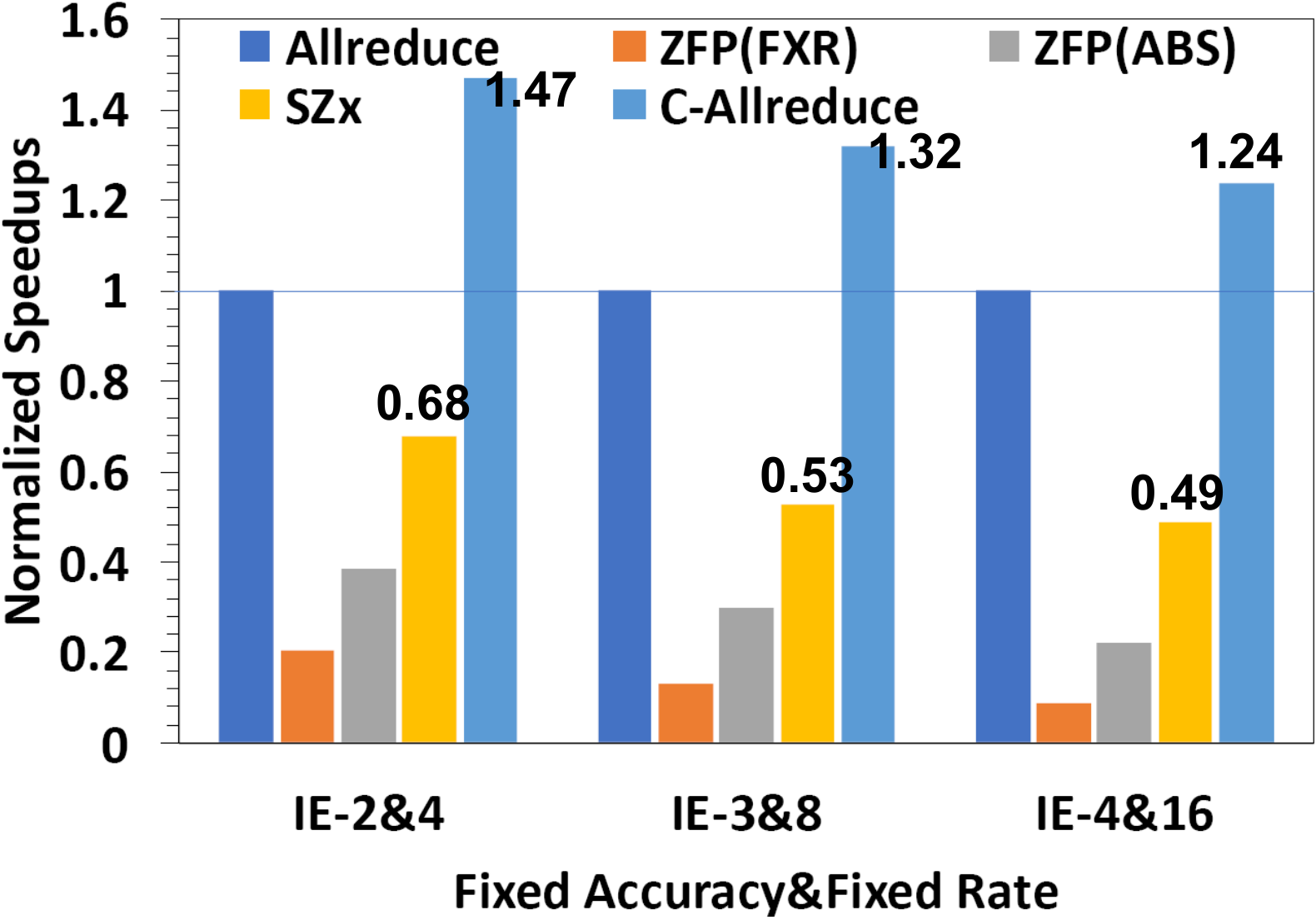}}
    \caption{Compare the image stacking performance of our C-Allreduce and multiple baselines in different error bounds for SZx and ZFP(ABS) and rates for ZFP(FXR).} 
    \label{fig-image-stacking-perf}
\end{figure}

We show the performance results and validate the high quality of the stacked images generated under our compression-integrated MPI collective framework, as shown in Figures \ref{fig-image-stacking-perf} and \ref{fig-image-stacking-quality}, respectively. For our C-Allreduce, we could see that, with the increase of error bounds, the performance drops but the quality of the reconstructed image rises. The highest speedup can be witnessed in the 1E-2 case, where our C-Allreduce has 1.5$\times$ higher performance compared to the original Allreduce. Nevertheless, the three compression-integrated baselines (i.e., SZx, ZFP(ABS), and ZFP(FXR)) all result in performance degradation compared with the original Allreduce, regardless of the absolute error bounds or rates. With an error bound of 1E-4, our C-Allreduce shows an excellent reconstructed image quality, whereas the ZFP(ABS) integrated baseline with the same error bound cannot achieve the same quality. This is because that our proposed frameworks can significantly decrease the error propagation of the reconstructed data, while the CPR-P2P cannot. Due to the page limitation, we do not show the reconstructed image of the SZx integrated baseline as it has even worse quality compared with ZFP(ABS). Besides, the ZFP(FXR) baseline with a rate of 4 has the worst reconstructed quality and the stacked image is completely different from the original one, as the fixed-rate mode cannot ensure a bounded accuracy. 

In addition to the visualization evaluation, we complement our findings with numerical metrics to further demonstrate the exceptional accuracy of our approach. For an error bound of 1E-2, our C-Allreduce method achieves a Peak Signal-to-Noise Ratio (PSNR) \cite{PSNR} of 42.86 and a Normalized Root Mean Square Error (NRMSE) \cite{shcherbakov2013errormetrics} of 7E-3, which illustrate a suboptimal data quality. However, when the error bound is tightened to 1E-3, we observe a notable increase in accuracy, with the PSNR soaring to 57.97 and the NRMSE dropping to 1E-3, which represent a great data quality. Consistent with the trends observed in image quality, the highest accuracy is attained at the 1E-4 error bound. At this threshold, our approach reaches a PSNR of 79.57 and maintains an NRMSE of 1E-4. These numerical metrics further validate the well-controlled accuracy of our method. In a nutshell, our C-Allreduce integrated with \textit{C-Coll} framework can remarkably increase the performance of the original Allreduce and also preserves the quality of the original datasets very well at the error bound of 1E-3 and 1E-4 (see Figure \ref{fig-image-stacking-quality} (e) and (f)).

\begin{figure}[ht]
    \centering
    \vspace{-2mm}
    \subfloat[Original Data]{\includegraphics[width=.35\linewidth]{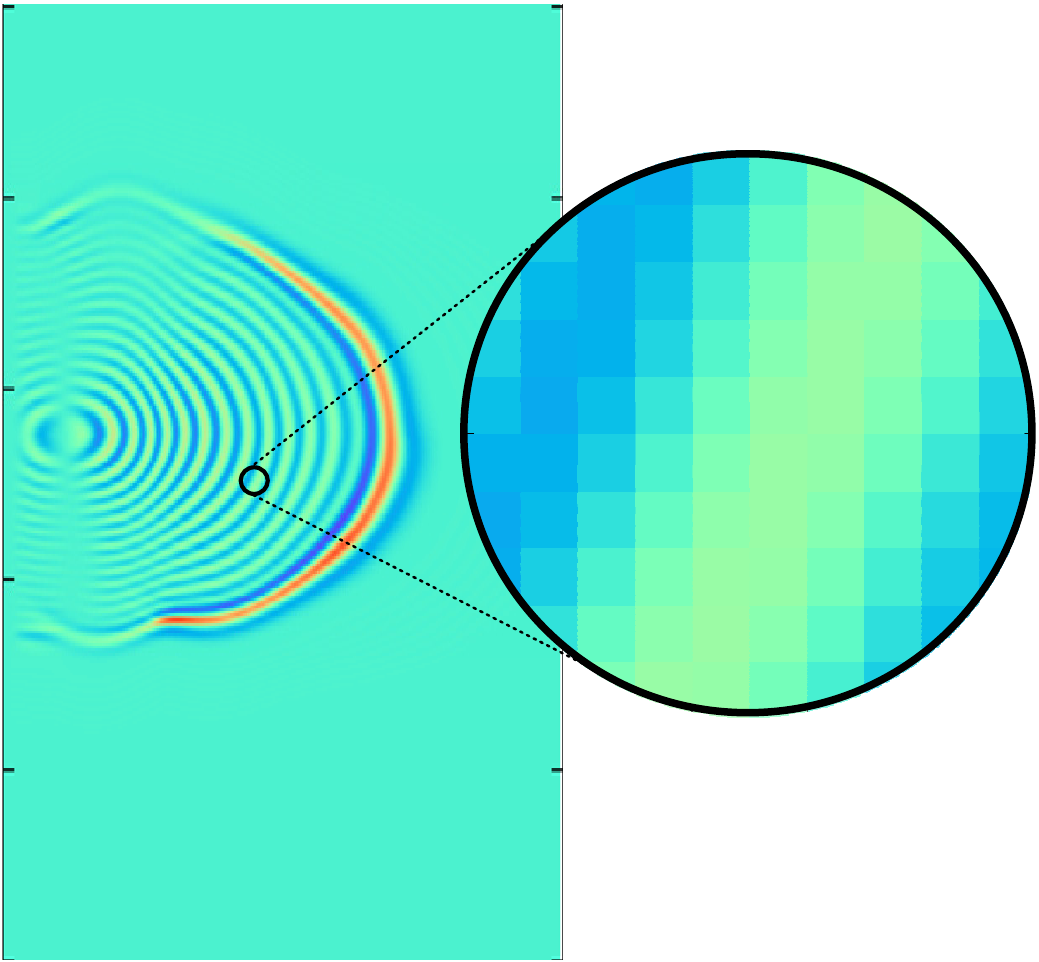}\label{ori-fig}}
    \subfloat[ZFP(FXR)]
    {\includegraphics[width=.35\linewidth]{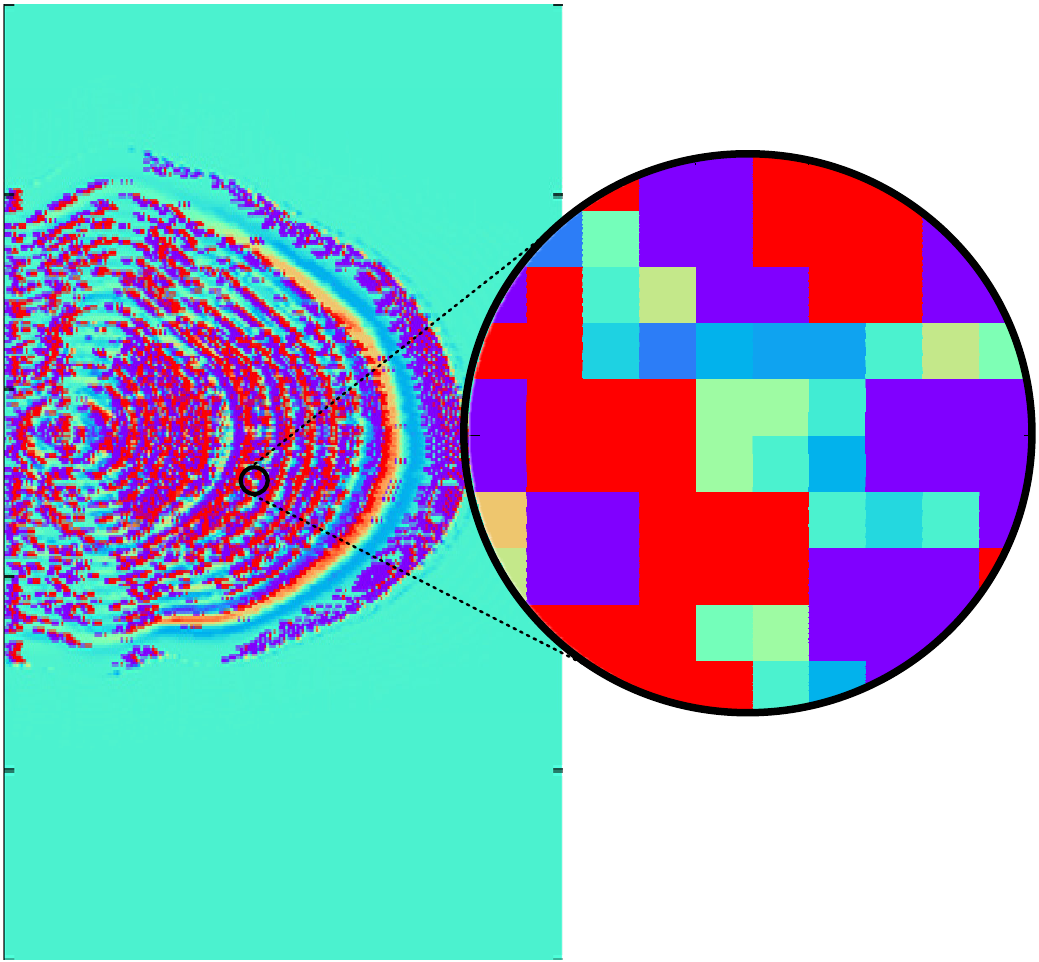}}

    \subfloat[ZFP(ABS)]{\includegraphics[width=.35\linewidth]{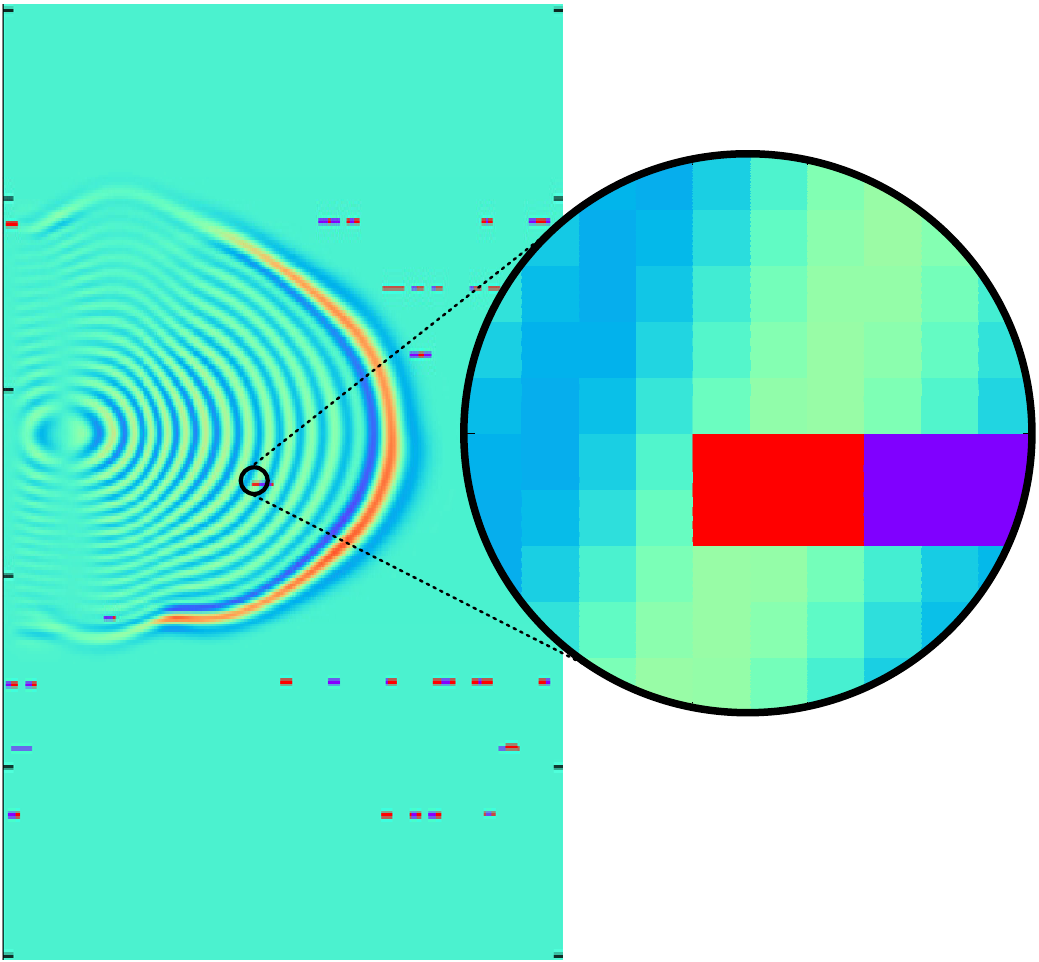}\label{ZFP(ABS)-fig}}
    \subfloat[C-Allreduce(1E-2)]{\includegraphics[width=.35\linewidth]{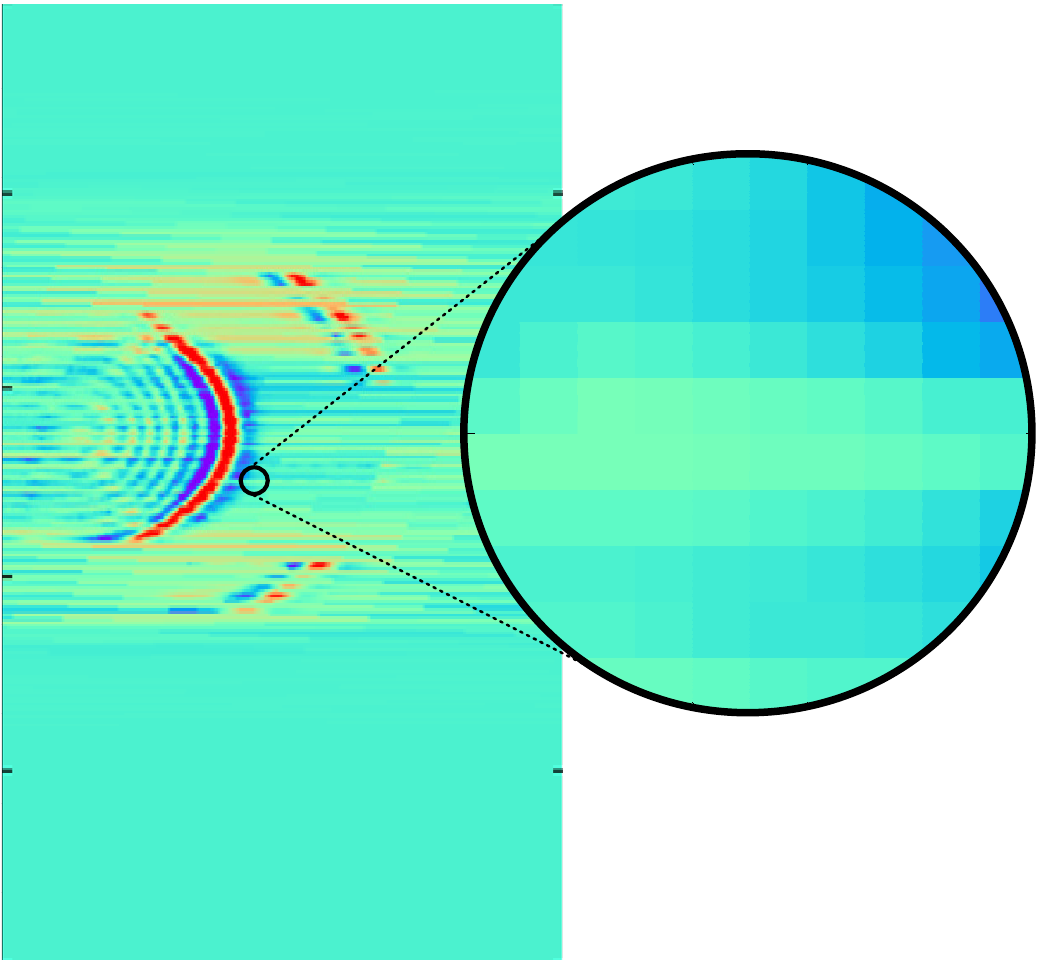}}

    \subfloat[C-Allreduce(1E-3)]{\includegraphics[width=.35\linewidth]{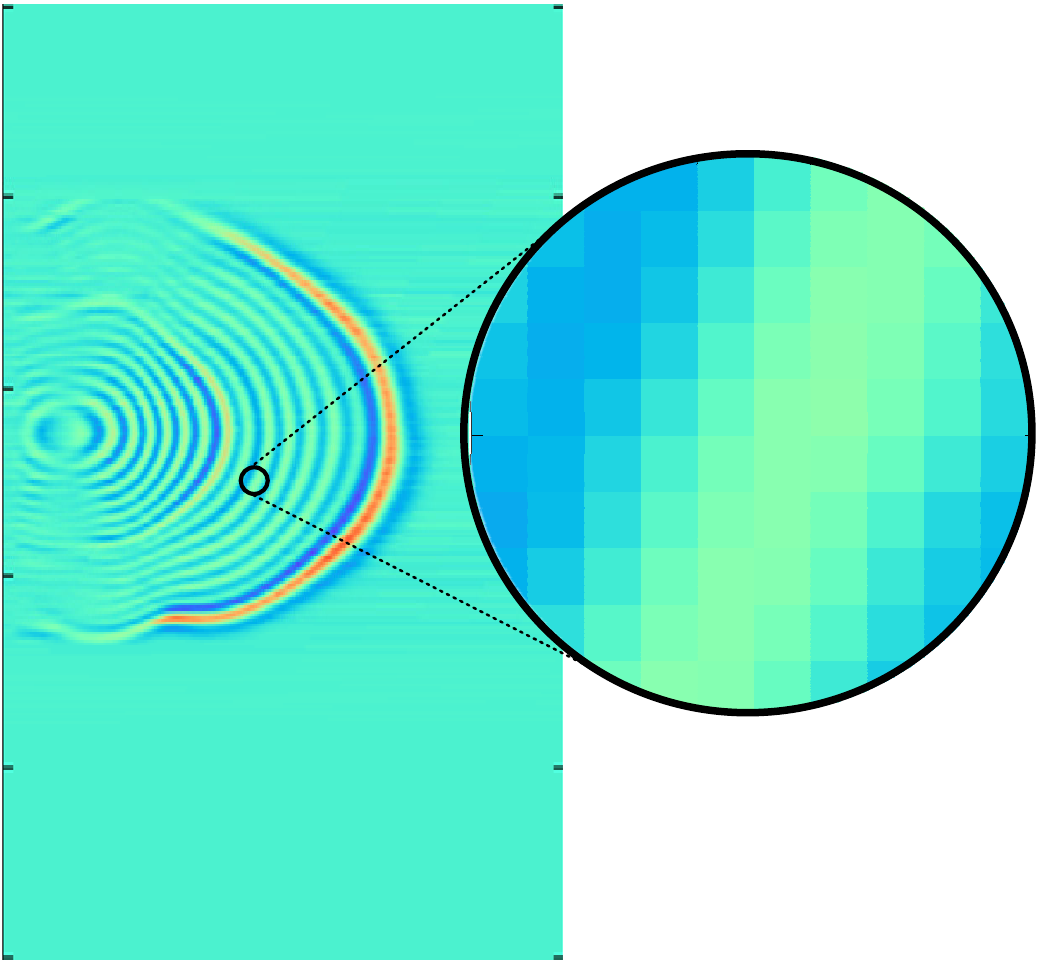}\label{SZX(1E-3)-fig}}
    \subfloat[C-Allreduce(1E-4)]{\includegraphics[width=.35\linewidth]{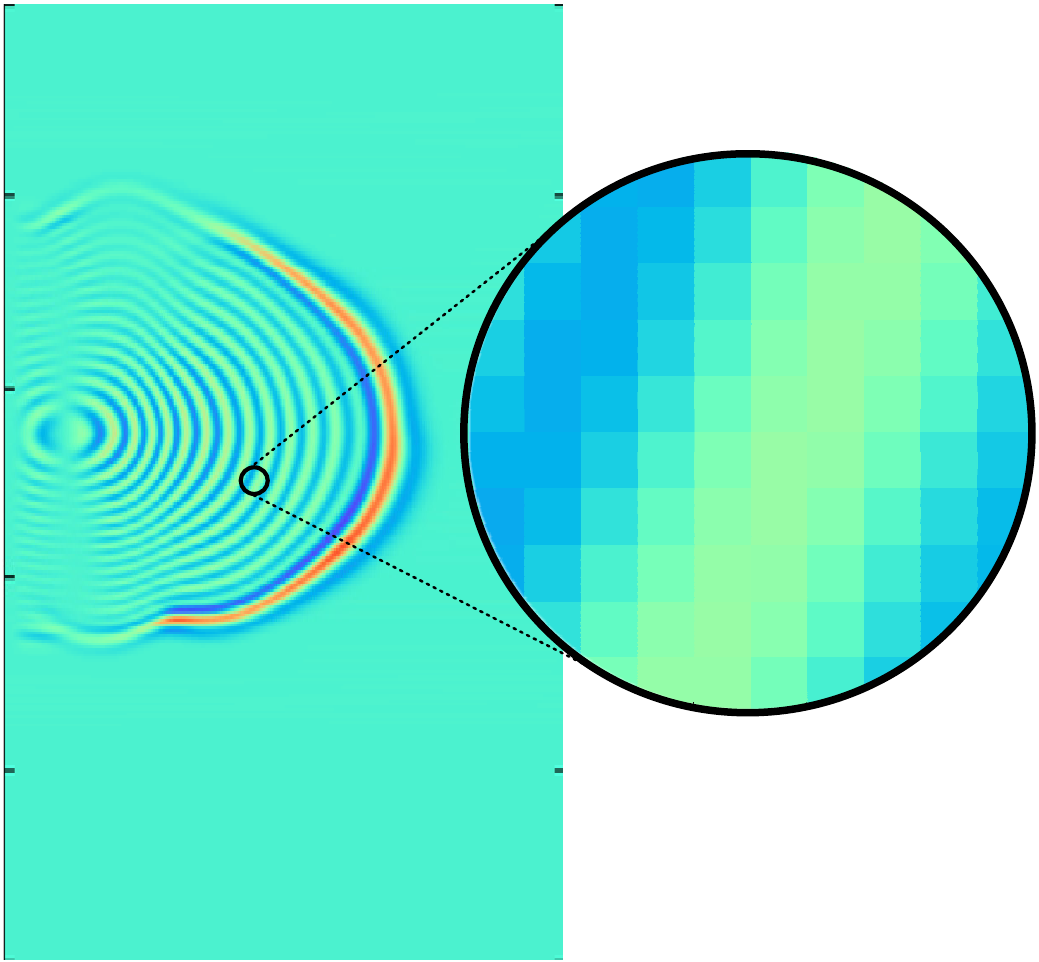}}  
    \caption{Compare the reconstructed image qualities of our C-Allreduce with different ZFP-integrated baselines.} 
    \label{fig-image-stacking-quality}
\end{figure}

\vspace{-2mm}

\section{Conclusion and Future Work}
\label{sec:conclusion}

In this paper, we introduce \textit{C-Coll}, a novel design for lossy-compression-integrated MPI collectives that significantly improves performance with bounded errors. Our two proposed high-performance frameworks for compression-integrated MPI collectives, together with customized pipe-lined SZx, enable us to implement C-Allreduce, which outperforms the original Allreduce by up to \textbf{2.1$\times$} while preserving high data quality. We demonstrate the generalizability of our approaches through C-Scatter and C-Bcast, which outperform the original MPI\_Scatter and MPI\_Bcast by up to \textbf{1.8$\times$} and \textbf{2.7$\times$}, respectively. In summary, our research has addressed the issues of sub-optimal performance, lack of generality, and unbounded errors in lossy-compression-integrated MPI collectives, laying the foundation for future research in this area. Moving forward, we plan to expand our research by implementing more \textit{C-Coll} based collectives and deploying our design on other hardware, such as GPUs and AI accelerators.

\section*{Acknowledgment}
This research was supported by the Exascale Computing Project (ECP), Project Number: 17-SC-20-SC, a collaborative effort of two DOE organizations – the Office of Science and the National Nuclear Security Administration, responsible for the planning and preparation of a capable exascale ecosystem, including software, applications, hardware, advanced system engineering and early testbed platforms, to support the nation’s exascale computing imperative. The material was supported by the U.S. Department of Energy, Office of Science, Advanced Scientific Computing Research (ASCR), under contract DE-AC02-06CH11357, and supported by the National Science Foundation under Grant OAC-2003709, OAC-2104023, and OAC-2311875. The experimental resource for this paper was provided by the
Laboratory Computing Resource Center on the Bebop cluster at Argonne National
Laboratory.


\end{document}